\documentclass[12 pt, reqno]{amsart}
\usepackage{amsthm,amsfonts,amssymb,euscript,cite}
\usepackage{graphics}
\usepackage{xcolor} 

\usepackage{graphicx}
\usepackage{color}
\usepackage{transparent}

\usepackage{fullpage} 

\definecolor{green}{rgb}{0,0.8,0} 

\newtheorem{theorem}{Theorem}[section]
\newtheorem{corollary}[theorem]{Corollary}

\newtheorem{proposition}[theorem]{Proposition}
\newtheorem{conjecture}[theorem]{Conjecture}
\theoremstyle{definition}

\theoremstyle{remark}
\newtheorem{remark}[theorem]{Remark}
\numberwithin{equation}{section}

\newcommand{\abs}[1]{\vert#1\vert}

\newcommand{\set}[1]{\{#1\}}

\newcommand{\ep}{\epsilon}
\def\beaa{\begin{eqnarray*}}
\def\eeaa{\end{eqnarray*}}
\def\bea{\begin{eqnarray}}
\def\eea{\end{eqnarray}}
\def\be{\begin{equation}}
\def\ee{\end{equation}}

\newcommand{\ud}{d}
\newcommand{\rd}{\partial}
\newcommand{\nb}{\nabla}

\newcommand{\bb}{\Big}

\newcommand{\alp}{\alpha}
\newcommand{\bt}{\beta}
\newcommand{\gmm}{\gamma}
\newcommand{\Gmm}{\Gamma}
\newcommand{\dlt}{\delta}

\newcommand{\eps}{\epsilon}

\newcommand{\kpp}{\kappa}

\newcommand{\Sgm}{\Sigma}
\newcommand{\tht}{\theta}

\newcommand{\omg}{\omega}

\newcommand{\Omg}{\Omega}

\newcommand{\bbN}{\mathbb N}

\newcommand{\bbR}{\mathbb R}
\newcommand{\bbS}{\mathbb S}

\newcommand{\calC}{\mathcal C}
\newcommand{\calD}{\mathcal D}

\newcommand{\calH}{\mathcal H}
\newcommand{\calI}{\mathcal I}

\newcommand{\uC}{\underline{C}}

\newcommand{\lapp}{{\Delta} \mkern-13mu /\,}

\newcommand{\f}{\frac}

\newcommand{\CH}{\calC \calH^{+}}		
\newcommand{\EH}{\calH^{+}}

\newcommand{\Ufree}{U_{0}}
\newcommand{\Ufixed}{\overline{U}_{0}}

\begin{document}

\title[]{Proof of linear instability of the Reissner-Nordstr\"om Cauchy horizon under scalar perturbations}
\author{Jonathan Luk}
\address{Department of Pure Mathematics and Mathematical Statistics, Cambridge University, Cambridge CB3 0WB, UK}
\email{jluk@dpmms.cam.ac.uk}

\author{Sung-Jin Oh}%
\address{Department of Mathematics, UC Berkeley, Berkeley, CA 94720, USA}%
\email{sjoh@math.berkeley.edu}%


\begin{abstract}
It has long been suggested that solutions to linear scalar wave equation
$$\Box_g\phi=0$$
on a fixed subextremal Reissner-Nordstr\"om spacetime with non-vanishing charge are generically singular at the Cauchy horizon. We prove that generic smooth and compactly supported initial data on a Cauchy hypersurface indeed give rise to solutions with infinite nondegenerate energy near the Cauchy horizon in the interior of the black hole. In particular, the solution generically does not belong to $W^{1,2}_{loc}$. This instability is related to the celebrated blue shift effect in the interior of the black hole. The problem is motivated by the strong cosmic censorship conjecture and it is expected that for the full nonlinear Einstein-Maxwell system, this instability leads to a singular Cauchy horizon for generic small perturbations of Reissner-Nordstr\"om spacetime. Moreover, in addition to the instability result, we also show as a consequence of the proof that Price's law decay is generically sharp along the event horizon.
\end{abstract}

\maketitle

\section{Introduction}
In this paper, we consider the linear wave equation 
\begin{equation}\label{wave.eqn}
\Box_g\phi=0
\end{equation}
on a subextremal Reissner-Nordstr\"om spacetime.  Here, $\Box_g$ is the Laplace-Beltrami operator associated to the metric $g$ of a subextremal Reissner-Nordstr\"om spacetime. In a local coordinate system, $g$ is given by
$$g=-(1-\f {2M} r+\f{e^2}{r^2})dt^2+(1-\f {2M} r+\f{e^2}{r^2})^{-1} dr^2+r^2 d\sigma_{\mathbb S^2},$$
where $d\sigma_{\mathbb S^2}$ denotes the standard metric on the $2$-sphere with radius $1$. Throughout this paper, we take $e$ and $M$ to be in the range
$$0<|e|<M,$$
i.e., the spacetime is subextremal with non-vanishing charge. 

It is known that the solutions to \eqref{wave.eqn} with sufficiently regular initial data decay with a polynomial rate\footnote{The polynomial decay is with respect to the $u$, $v$ coordinates defined in Section \ref{sec.geometry}.} in the exterior region of the spacetime. Moreover, it is shown recently by Franzen \cite{Fra} that the solution remains uniformly bounded everywhere in the spacetime including in the black hole region and up to the Cauchy horizon. We refer the readers to Sections \ref{sec.exterior} and \ref{SCC} for a further discussion of these results.

On the other hand, it is expected that generically the derivative of the solution $\phi$ with respect to a regular vector field transversal to the Cauchy horizon is singular. Our main result in this paper shows that this is indeed the case. Let the initial hypersurface $\Sgm_{0}$ be a complete $2$-ended asymptotically flat Cauchy hypersurface for the maximal globally hyperbolic development of a subextremal Reissner-Nordstr\"om spacetime with non-vanishing charge as depicted in the Penrose diagram in Figure~\ref{fig:main.thm} (see Section~\ref{sec.geometry} below for the notation). The following is the first version of the main theorem:
\begin{theorem}[Main theorem, first version]\label{main.thm}
Generic smooth and compactly supported initial data to \eqref{wave.eqn} on $\Sgm_{0}$ give rise to solutions that are not in $W^{1,2}_{loc}$ in a neighborhood of any point on the future Cauchy horizon $\CH$.
\end{theorem}
\begin{figure}[h]
\begin{center}
\def\svgwidth{200px}
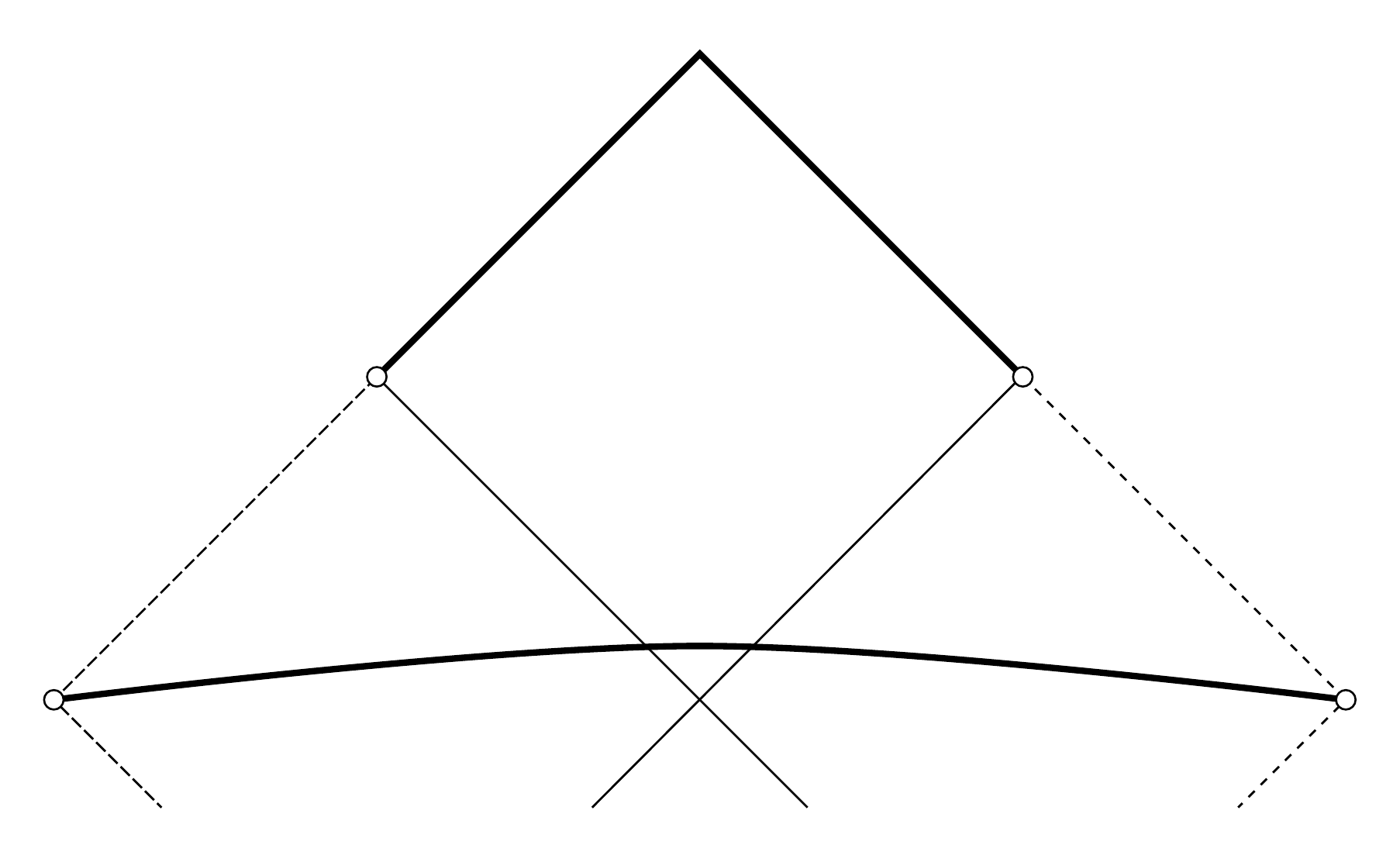 
\caption{} \label{fig:main.thm}
\end{center}
\end{figure}

This theorem is motivated by the celebrated strong cosmic censorship conjecture in general relativity (see Section \ref{SCC}). A particular consequence of this conjecture is that small perturbations to the Reissner-Nordstr\"om spacetime for the \emph{nonlinear} Einstein-Maxwell system leads to a Cauchy horizon that is singular\footnote{In fact, a priori one may even conjecture that a stronger singularity forms and the spacetime does not contain a Cauchy horizon at all. However, the recent work of Dafermos-Luk \cite{DL} suggests that this is not the case (see discussions in Section \ref{SCC}).}. A more precise version of the conjecture in this setting has been proposed by Christodoulou \cite{Chr}, suggesting that generically the metric is inextendible in $W^{1,2}_{loc}$ beyond the Cauchy horizon. While this conjecture remains open, our present paper initiates the study of the instability mechanism by proving a generic blow up result in $W^{1,2}_{loc}$ for the linear wave equation, which can be viewed as a ``poor man's'' version of the linearization for the Einstein-Maxwell system.

Theorem \ref{main.thm} is proved via considering the spherically symmetric part of the solution. We show that given any regular solution, the spherically symmetric part of the data can be perturbed to guarantee that the new solution is singular. This is sufficient to guarantee that the set of general (i.e., not necessarily spherically symmetric) data giving rise to regular solutions at the Cauchy horizon has co-dimension at least $1$.  

The instability result for the spherically symmetric part of the solution is proved via identifying a condition near null infinity (see \eqref{L.condition} below) which guarantees that the solution does not belong to $W^{1,2}_{loc}$ near the Cauchy horizon. The proof proceeds via a contradiction argument in which we show that if the solution is regular in the interior of the black hole, then we can prove upper bounds for the solution that is too strong and will contradict a lower bound for the solution that follows from the condition \eqref{L.condition}.

As a consequence of the proof of Theorem \ref{main.thm}, we also show that the Price's law bound for $\rd_t\phi$ along the event horizon is sharp. We summarize the result here and refer the readers to Corollary \ref{linear.Price.law.sharp} for a more precise formulation:
\begin{theorem}
The Price's law decay for $\rd_t\phi$ along the event horizon is generically sharp.
\end{theorem}

We will give a more precise version of the main theorem (see Corollary \ref{main.cor}) and will further discuss the method of the proof in Section \ref{sec.main.thm}. Before that, we first give a brief introduction to the geometry of the Reissner-Nordstr\"om spacetime.

\subsection{Geometry of Reissner-Nordstr\"om}\label{sec.geometry}
Reissner-Nordstr\"om spacetimes are the unique $2$-parameter family of static, spherically symmetric solutions to the Einstein-Maxwell system. This family is parametrized by the mass $M$ and the charge $e$ of the black hole. In the parameter range $0<|e|<M$, the geometry of the maximal globally hyperbolic development of Reissner-Nordstr\"om data on a complete Cauchy hypersurface $\Sgm_{0}$ with two asymptotically flat ends (denoted $i^{0}$) is depicted in the Penrose diagram below in Figure~\ref{fig:RN}.

\begin{figure}[h]
\begin{center}
\def\svgwidth{200px}
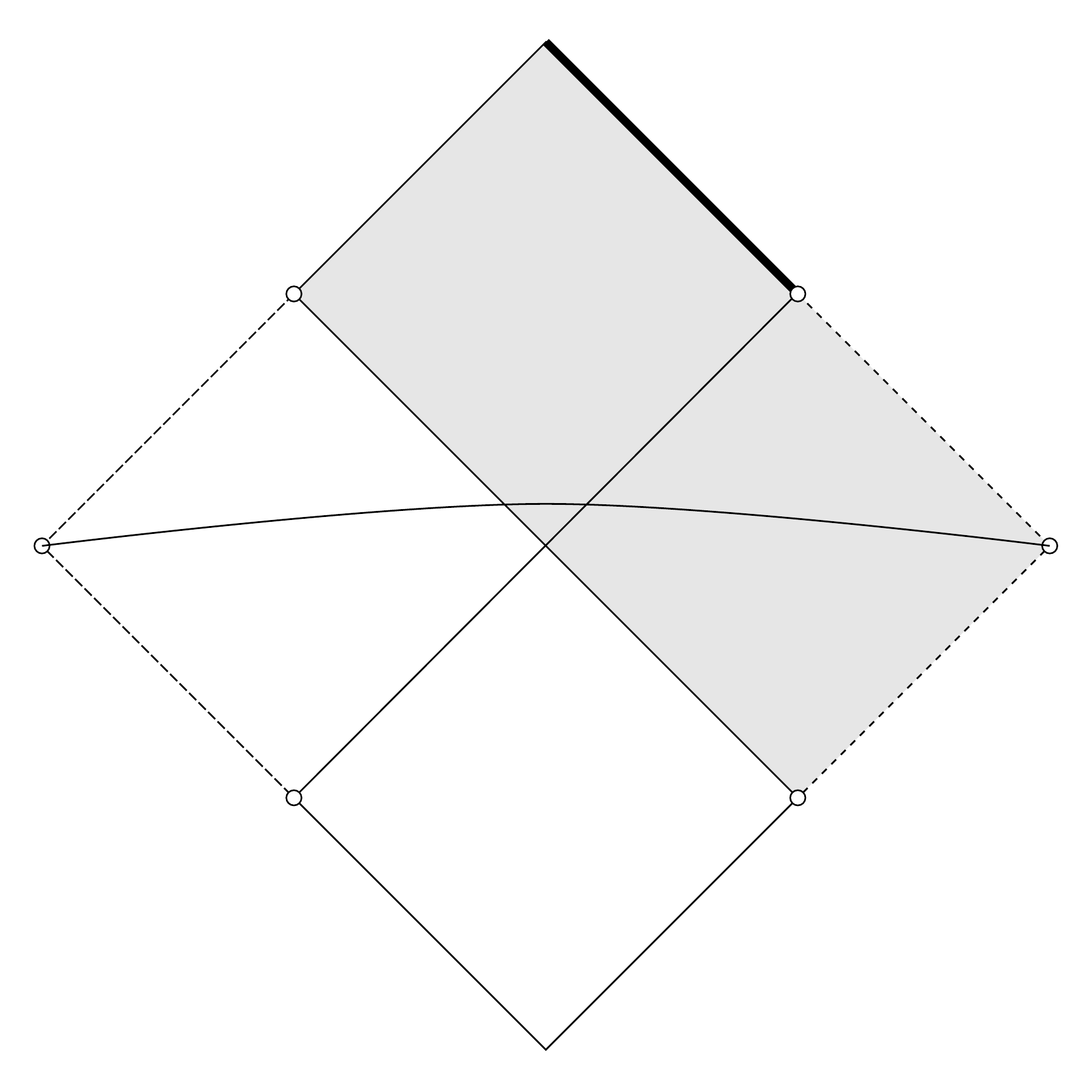 
\caption{} \label{fig:RN}
\end{center}
\end{figure}

As mentioned before, the metric of Reissner-Nordstr\"om can be given in a local coordinate chart by
\begin{equation} \label{eq:RNmet}
g=-(1-\f {2M} r+\f{e^2}{r^2})dt^2+(1-\f {2M} r+\f{e^2}{r^2})^{-1} dr^2+r^2 d\sigma_{\mathbb S^2}.
\end{equation}
The Reissner-Nordstr\"om spacetime has a black hole region, labeled as $II$ in Figure~\ref{fig:RN}, which is separated from the exterior regions $I$ and $I'$ by the event horizon $\{r=r_+\}$, where $r_+$ is the larger root of $r^2-2Mr+e^2$. This black hole region is characterized by the feature that every point in it is not contained in the past of future null infinity $\calI^{+}$, i.e., no signal can be sent from the black hole region to null infinity. The bifurcate null hypersurface $\EH$ dividing the black hole and the two exterior regions is called the \emph{future event horizon}. The other interior region $II'$ (sometimes called a white hole region), past null infinity $\calI^{-}$ and past event horizon $\calH^{-}$ are defined similarly by reversing time. The expression \eqref{eq:RNmet} for the Reissner-Nordstr\"om metric exhibits a coordinate singularity along $\EH$ and $\calH^{-}$, but is valid everywhere else, i.e., on each of the regions $I$, $II$, $I'$, $II'$.

For the purpose of this paper, the most relevant feature of the Reissner-Nordstr\"om spacetime is that it has a smooth Cauchy horizon $\{r=r_-\}$, where $r_-$ is the smaller root of $r^2-2Mr+e^2$. The component of $\set{r = r_{-}}$ to the future of $II$, denoted $\CH$ in Figure~\ref{fig:RN}, is called the \emph{future Cauchy horizon}; the past Cauchy horizon $\mathcal{CH}^{-}$ is defined similarly by time reversal. The presence of the smooth Cauchy horizon in particular allows the maximal globally hyperbolic development to be extended smoothly and non-uniquely (!) as solutions to the Einstein-Maxwell system. It is precisely this feature that is expected to be non-generic according to the strong cosmic censorship conjecture (see Section \ref{SCC}).

In view of the symmetry between the two asymptotically flat ends, we henceforth consider only a subset of the maximal globally hyperbolic development of Reissner-Nordstr\"om, namely, the shaded region in Figure~\ref{fig:RN}. We moreover restrict our attention to the ``incoming'' portion of the bifurcate future Cauchy horizon, which is the part to the right of the shaded region in Figure~\ref{fig:RN} (drawn with a bold line). Once we obtain the generic blow up result in this region, we can also derive an analogous generic blow up result near every point on the whole future Cauchy horizon for generic initial data on a complete $2$-ended asymptotically flat initial Cauchy hypersurface $\Sgm_{0}$ by repeating the same argument. 
We will frequently restrict our attention to an even smaller subset, namely, the shaded region in Figure~\ref{fig:data} (see Section~\ref{sec.main.thm} below). It can be easily shown that the blow up result on the intersection of $\CH$ with this region implies a blow up result on the whole ``incoming'' portion of $\CH$ by a simple symmetry argument (see the proof of Corollary~\ref{main.cor} and Remark \ref{rem:wholeCH}).


The shaded region in Figure~\ref{fig:RN} contains both an exterior region $I$ and an interior region $II$ of the black hole. In the next two sections, we will define the coordinates that we use in each of these regions.

\subsubsection{Coordinates in the exterior region}
We define null coordinates in the black hole exterior (region $I$ in Figure~\ref{fig:RN}) as follows.
Let
$$r^*=r+(M+\frac{2M^2-e^2}{2\sqrt{M^2-e^2}})\log (r-r_+) +(M-\frac{2M^2-e^2}{2\sqrt{M^2-e^2}})\log (r-r_-).$$
Let also 
$$v=\frac 12(t+r^*),\quad u=\frac 12(t-r^*).$$
As a consequence
$$\frac{\rd}{\rd v}= \frac{\rd}{\rd r^*}+\frac{\rd}{\rd t},\quad\frac{\rd}{\rd u}= \frac{\rd}{\rd t}-\frac{\rd}{\rd r^*}.$$
In the coordinate system $(u, v, \tht, \varphi)$, where $(\tht, \varphi)$ is a spherical coordinate system on $\bbS^{2}$, the Reissner-Nordstr\"om metric in the exterior region takes the form 
$$g=-4\Omg^2 du\,dv+ r^2 d\sigma_{\mathbb S^2}.$$
where $\Omg^2=(1-\f{2M}{r}+\f{e^2}{r^2})$. Also
\begin{equation}\label{ext.r}
\rd_v r=-\rd_u r=1-\f{2M}{r}+\f{e^2}{r^2}.
\end{equation}
In this coordinate system, the future null infinity $\calI^{+}$ corresponds to the limit $\set{v = \infty}$, whereas the future event horizon $\EH$ is represented by $\set{u = \infty}$.

\subsubsection{Coordinates in the interior region}
We now turn to the black hole interior, which is labeled as $II$ in Figure~\ref{fig:RN}. Let 
$$r^*=r+(M+\frac{2M^2-e^2}{2\sqrt{M^2-e^2}})\log (r-r_+) +(M-\frac{2M^2-e^2}{2\sqrt{M^2-e^2}})\log (r-r_-).$$
Define the null coordinates
$$v=\frac 12(r^*+t),\quad u=\frac 12(r^*-t),$$
which implies
$$\frac{\rd}{\rd v}= \frac{\rd}{\rd r^*}+\frac{\rd}{\rd t},\quad\frac{\rd}{\rd u}= \frac{\rd}{\rd r^*}-\frac{\rd}{\rd t}.$$
In the coordinate system $(u, v, \tht, \varphi)$, the Reissner-Nordstr\"{o}m metric again takes the form
$$g=-4\Omg^2 du\,dv+ r^2 d\sigma_{\mathbb S^2}$$
except that we now have $\Omg^2=-(1-\frac{2M}{r}+\frac{e^2}{r^2})$ instead. Moreover, in this region
\begin{equation}\label{int.r}
\rd_v r=\rd_u r=1-\f{2M}{r}+\f{e^2}{r^2}.
\end{equation}
In this coordinate system, the future event horizon $\EH$ corresponds to the limit $\set{u = -\infty}$, whereas the ``incoming'' portion of the future Cauchy horizon $\CH$ is represented by $\set{v = \infty}$.

\subsubsection{The wave equation}
In both the exterior and the interior region, the wave equation takes the form
$$\rd_u\rd_v\phi=-\f{\rd_v r \rd_u\phi}{r}-\f{\rd_u r \rd_v\phi}{r}+\f{\Omg^2 \lapp\phi}{r^2},$$
where $\lapp$ denotes that Laplace-Betrami operator on the standard $2$-sphere with radius $1$. 
In most of this paper, we will be particularly concerned with this equation in spherical symmetry and the equation takes the form
$$\rd_u\rd_v\phi=-\f{\rd_v r \rd_u\phi}{r}-\f{\rd_u r \rd_v\phi}{r}.$$
In other words, in view of \eqref{ext.r} and \eqref{int.r}, the wave equation in spherical symmetry can be expressed as
\begin{equation}\label{wave.eqn.ext}
\rd_u\rd_v\phi=\f{\Omg^2}{r}(\rd_v\phi-\rd_u\phi)\quad\mbox{in the exterior region},
\end{equation}
and
\begin{equation}\label{wave.eqn.int}
\rd_u\rd_v\phi=\f{\Omg^2}{r}(\rd_v\phi+\rd_u\phi)\quad\mbox{in the interior region}.
\end{equation}
Moreover, we frequently find it convenient to write the equation in the following equivalent form
\begin{equation}\label{re.wave.eqn.ext}
\rd_u\rd_v (r\phi)=-\f{2(M-\f{e^2}{r})\Omg^2\phi}{r^2}
\end{equation}
in the exterior region.

\subsubsection{Nondegenerate energy}\label{sec.nondegenerate}
We define the nondegenerate energy associate to the solution to the wave equation. To this end, we define a globally regular timelike vector field $N$ as follows: Let
\begin{equation}
N=(1+\f{\chi_{N,1}(r)}{\Omg^2})\rd_u+(1+\f{\chi_{N,2}(r)}{\Omg^2})\rd_v
\end{equation}
where $\chi_{N,1}(r)$ and $\chi_{N,2}(r)$ are smooth cutoff functions such that 
\begin{equation*}
\chi_{N,1}(r)=
\begin{cases}
1 &\mbox{for }r_+-\f{r_+-r_-}{8}\leq r\leq r_++\f{r_+-r_-}{8}\\
0 &\mbox{for }r\leq r_+-\f{r_+-r_-}{4}\mbox{ or }r\geq r_++\f{r_+-r_-}{4}
\end{cases}
\end{equation*}
and
\begin{equation*}
\chi_{N,2}(r) =
\begin{cases}
1 &\mbox{for }r_-\leq r\leq r_-+\f{r_+-r_-}{8}\\
0 &\mbox{for }r\geq r_-+\f{r_+-r_-}{4},
\end{cases}
\end{equation*}
subject to the condition $0\leq \chi_{N,1}(r), \chi_{N,2}(r)\leq 1$ everywhere.

It is easy to check that $N$ is timelike and moreover is regular near both the event horizon and the Cauchy horizon. Using the vector field $N$, we can define the nondegenerate energy as follows:
Let 
$$T_{\mu\nu}=\rd_\mu\phi\rd_\nu\phi - \f12 g_{\mu\nu}(g^{-1})^{\alp\beta}\rd_\alp\phi\rd_{\beta}\phi$$
and given any spacelike or null hypersurface $\Sigma$, we define the nondegenerate energy to be
$$E_{\Sigma}(\phi):=\int_{\Sigma} T_{\mu\nu} N^\mu n^{\nu}_\Sigma.$$
When $\Sigma$ is spacelike, $n^{\nu}_\Sigma$ in the above expression is the unit normal to $\Sigma$ and the integration is with respect to the volume form associated to the induced metric on $\Sigma$. In the case where $\Sigma$ is a null hypersurface, there is no canonical volume form, but we can fix the normalization by requiring that we integrate over the $3$-form $\omega$ such that 
$n_{\Sigma}^{\flat} \wedge \omega$ is the space-time volume form, where $n_{\Sgm}^{\flat}$ is the metric dual of $n_{\Sgm}$.

In proving our main theorem (Theorem \ref{main.thm}), we will in fact show that for generic initial data, the nondegenerate energy on a null hypersurface transversal to the Cauchy horizon is infinite. This will in particular show that the solution is not in $W^{1,2}_{loc}$, as claimed in Theorem \ref{main.thm}.

\subsubsection{Notation}\label{sec.notation}
We end this section with some notation for various subsets of the Reissner-Nordstr\"om spacetime. 

We use $C_u$ to denote a constant $u$ hypersurface and $\underline C_v$ to denote a constant $v$ hypersurface. We will frequently consider a constant $v$ hypersurface that penetrates the event horizon and we use $\underline C^{int}_v$ and $\underline C^{ext}_v$ to denote the parts of the hypersurface in the interior and exterior region of the black hole respectively. Whenever there is no danger of confusion, we will drop the superscript in $\underline C$.

Our argument will involve constant $r$-hypersurfaces, which we denote by $\gamma_R:=\{r=R\}$. We will also abuse notation to denote the same set as $\gamma_{R^{\ast}}:=\{r^{\ast}=R^{\ast}\}$. Here, $R^{\ast}$ denotes the $r^{\ast}$ value corresponding to $R$. We will also use $u_{r^{\ast}}(v)$ to denote the unique value of $u$ such that $r^{\ast}(u_{r^{\ast}},v)=r^{\ast}$ and similarly for $v_{r^{\ast}}(u)$.

We introduce the following convention for integration. On $C_u$ (resp. $\underline C_v$), we use the convention that the integration is always with respect to $dv$ (resp. $du$).
On the constant $r$-hypersurface $\gamma_r$, unless otherwise specified, we parametrize the curve by the spacetime coordinate $v$ and integrate with respect to $dv$. 
On the other hand, in a spacetime region, the integration is with respect to $du\, dv$. (Notice that this is not equal to the integration with respect to the volume form induced by the spacetime metric!)

Finally, we note that in Section \ref{sec.interior}, we find it convenient to relabel the hypersurfaces and use a different set of notation. We refer the readers to the beginning of Section \ref{sec.interior}.

\subsection{Statement of main theorem and outline of the proof}\label{sec.main.thm}
As mentioned earlier, to achieve the main theorem, it suffices to restrict to spherically symmetric solutions. This is because the Reissner-Nordstr\"om spacetime is spherically symmetric and we can decompose the solution into spherical harmonics. The blow up of the spherically symmetric mode in $W^{1,2}_{loc}$ then implies the blow up of the full solution by the orthogonality of the spherical harmonics.

We now describe the class of data that we consider in the paper. The initial data will be given on two transversely intersecting null hypersurfaces $\underline C_1$ and $C_{-U_0}$ as shown in Figure~\ref{fig:data}, where $\underline C_1$ is a horizon penetrating null hypersurface composed of $\uC_{1}^{ext} \cup \uC_{1}^{int}$. In the interior region, we only consider $\uC_{1}^{int}$ up to $u \leq -1$.
\begin{figure}[h]
\begin{center}
\def\svgwidth{200px}
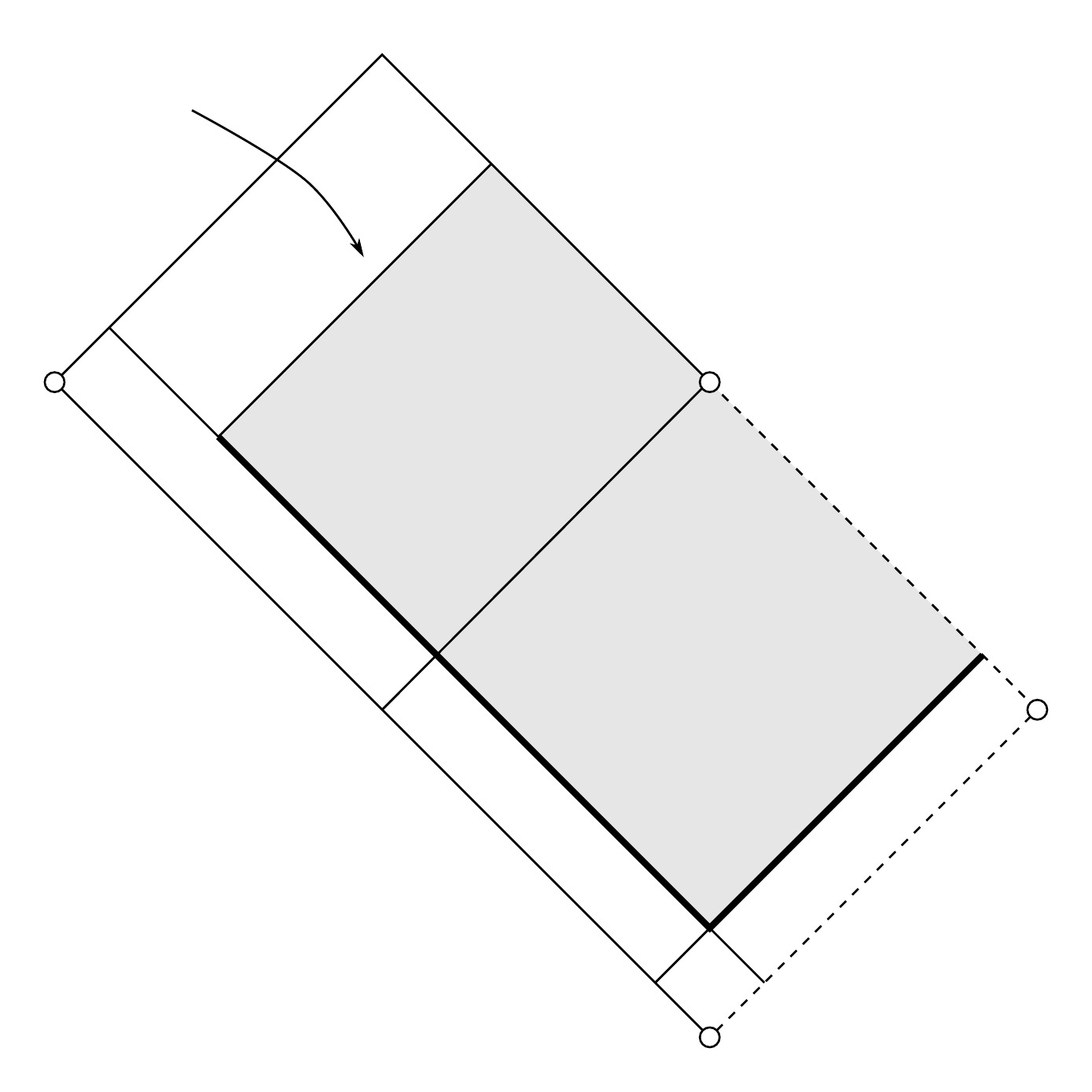 
\caption{} \label{fig:data}
\end{center}
\end{figure}

The shaded region in Figure~\ref{fig:data} corresponds to the domain of dependence of these hypersurfaces, which equals $\set{(u, v) : u \leq -1, v \leq 1}$ in the interior and $\set{(u, v) : u \geq -U_{0}, v \leq 1}$ in the exterior.

We require that for some $D>0$, the data on $C_{-U_0}\cap \set{v\geq 1}$ verify
\begin{equation}\label{data.1}
\sup_{C_{-U_0} \cap \set{v \geq 1}} r^{2} \abs{\frac{\rd_{v} \phi}{\rd_{v} r}} \leq D, \quad
	\sup_{C_{-U_0} \cap \set{v \geq 1}} r^3 \abs{\frac{\rd_{v}(r \phi)}{\rd_{v} r}} \leq D
\end{equation}
and the data on ${(\uC_{1}^{ext} \cap \set{u \geq -U_{0}})\cup(\uC_{1}^{int} \cap \set{u \leq -1})}$ obey
\begin{equation} \label{data.2}
	\sup_{(\uC_{1}^{ext} \cap \set{u \geq -U_{0}})\cup(\uC_{1}^{int} \cap \set{u \leq -1})} \abs{\frac{\rd_{u} \phi}{\rd_{u} r}} \leq D, \quad
	\sup_{(\uC_{1}^{ext} \cap \set{u \geq -U_{0}})\cup(\uC_{1}^{int} \cap \set{u \leq -1})} \abs{\frac{\rd_{u}(r \phi)}{\rd_{u} r}} \leq D.
\end{equation}
Assume moreover that $r^3\rd_v(r\phi)(-U_0,v)$ has a limiting value, i.e.,
\begin{equation}\label{data.3}
\lim_{v\to\infty}r^3\rd_v(r\phi)(-U_0,v)\mbox{ exists}.
\end{equation}
Notice that a particular consequence of the assumptions \eqref{data.1} and \eqref{data.2} is that the initial data have finite nondegenerate energy. In part of the paper, we will also restrict to compactly supported initial data\footnote{Notice that by the finite speed of propagation, we can indeed think of such solutions as arising from compactly supported initial data on a Cauchy hypersurface (see the proof of Corollary \ref{main.cor}).}, i.e., instead of \eqref{data.1}, we require $\phi$ and $\rd_v\phi$ to vanish identically on $C_{-U_0}\cap \set{v\geq 1}$. We will not make this stronger assumption for the main theorem.

We prove the following result for the wave equation in spherical symmetry, which shows that the derivative of $\phi$ is singular as long as a certain inequality holds along null infinity:
\begin{theorem}\label{decay.thm}
Let $\phi$ be a solution to \eqref{wave.eqn} with spherically symmetric initial data satisfying \eqref{data.1}-\eqref{data.3}. Assume that
\begin{equation}\label{L.condition}
	\mathfrak{L} := \lim_{v \to \infty} r^{3} \rd_{v} (r \phi)(-U_0, v) - \int_{-U_0}^{\infty} 2 M \Phi(u)  \, \ud u \neq 0,
\end{equation}
where\footnote{The fact that this limit exists is an easy consequence of the results in \cite{DRPL}.} $\Phi(u) := \lim_{v \to \infty} r \phi(u,v)$. Then, near the Cauchy horizon in the interior of the black hole, we have
\begin{equation}\label{blow.up.bd}
\int_{1}^\infty \log^{\alpha_0}(\f 1{\Omega}) (\rd_v\phi)^2(u,v)\, dv=\infty
\end{equation}
for every $u\in (-\infty,\infty)$ and every integer $\alpha_0>7$.
\end{theorem}

\begin{remark}
Notice that the coordinate system $(u,v)$ is non-regular near the Cauchy horizon. The nondegenerate energy on a constant $u$ null hypersurface is equivalent to
$$\int_1^\infty \Omg^{-2} (\rd_v\phi)^2 (u,v)\, dv.$$
Therefore, \eqref{blow.up.bd} indeed implies that the solution has infinite nondegenerate energy. Moreover, in the $(u,v)$ coordinate system, the $W^{1,2}_{loc}$ norm is given by
$$\int_{\mathcal U} \Omega^{-2}(\rd_v\phi)^2+\Omega^2(\rd_u\phi)^2 \,du\,dv,$$
where $\mathcal U$ is a small neighborhood of a point on the Cauchy horizon. Therefore, since the nondegenerate energy blows up on all constant $u$ null hypersurfaces, the $W^{1,2}_{loc}$ norm is also infinite.
\end{remark}

In order to apply Theorem \ref{decay.thm}, we also construct solutions satisfying the assumption of Theorem \ref{decay.thm}:
\begin{theorem}\label{construction.existence.thm}
For $U_{0} > 1$ sufficiently large, 
there exists a spherically symmetric solution $\phi_{sing}$ to \eqref{wave.eqn} with smooth and compactly supported initial data on $\uC_{1}$ and zero data on $C_{-U_{0}}$ such that 
$$\mathfrak{L}\neq 0.$$
In fact, the support of the initial data $\phi_{sing} \restriction_{\uC_{1}}$ is contained in $\uC_{1}^{ext} \cap \set{-U_{0} \leq u \leq -U_{0}+1}$.
\end{theorem}
Theorems \ref{decay.thm} and \ref{construction.existence.thm} easily imply that there exist spherically symmetric solutions to \eqref{wave.eqn} which are not in $W^{1,2}_{loc}$ near the Cauchy horizon. Moreover, one can conclude the following more precisely formulated version of Theorem \ref{main.thm}:
\begin{corollary}[Main theorem, second version]\label{main.cor}
Let $\Sgm_{0}$ be a complete $2$-ended asymptotically flat Cauchy hypersurface for a subextremal Reissner-Nordstr\"om spacetime with non-vanishing charge.
The set of smooth and compactly supported initial data on $\Sgm_{0}$ which lead to solutions with finite nondegenerate energy near the future Cauchy horizon $\CH$ has co-dimension at least $1$.
\end{corollary}
\begin{proof}
First, we claim that for sufficiently large $U_{0}$, the solution $\phi_{sing}$ given by Theorem~\ref{construction.existence.thm}, which is initially defined\footnote{To be pedantic, $\phi_{sing}$ is only defined in the shaded region in Figure~\ref{fig:data} according to Theorem~\ref{construction.existence.thm}. We can however prescribe the zero initial data on the whole $\uC_{-1}^{int}$, which trivially extends $\phi_{sing}$ to the domain of dependence of $\uC_{1} \cup C_{-U_{0}}$.} only on the domain of dependence of $\uC_{1} \cup C_{-U_{0}}$, extends to a smooth solution to the linear wave equation in the whole spacetime with compactly supported data on $\Sgm_{0}$. By finite speed of propagation, it suffices to prove this property for a particular Cauchy hypersurface. For convenience, we choose $\Sgm_{0}$ to be spherically symmetric and asymptotic to the $\set{t=0}$ hypersurface near each end as in Figure~\ref{fig:RN}. We take $U_{0}$ sufficiently large so that the segment $\uC_{1}^{ext} \cap \set{-U_{0} \leq u \leq -U_{0}+1}$ lies in the past of $\Sgm_{0}$. Consider the data $(f, g)$ on $\Sgm_{0}$ defined as
\begin{equation*}
	(f, g) (p) = \left\{
	\begin{array}{ll}
	(\phi_{sing}, n_{\Sgm_{0}} \phi_{sing})(p) & \hbox{ if $p$ is in the domain of dependence of $\uC_{1} \cup C_{-U_{0}}$},  \\
	(0, 0) & \hbox{ otherwise.} 
	\end{array}
	\right.
\end{equation*}
Note that $(f, g)$ is compactly supported in $\Sgm_{0}$, since the Cauchy hypersurface $\Sgm_{0}$ necessarily exits the domain of dependence of $\uC_{1} \cup C_{-U_{0}}$ near each end. It is also straightforward to see that $(f, g)$ is smooth.

Let $\phi'$ be the solution to the linear wave equation with $(\phi', n_{\Sgm_{0}} \phi') \restriction_{\Sgm_{0}} = (f, g)$. Using finite speed of propagation and the support property of the initial data for $\phi_{sing}$ in Theorem~\ref{construction.existence.thm}, we see that $\phi'$ agrees with $\phi_{sing}$ in the domain of dependence of $\uC_{1} \cup C_{-U_{0}}$. Therefore, $\phi'$ is the desired extension of $\phi_{sing}$. Note furthermore that $\phi'$ is identically zero in the other exterior region $I'$ (see Figure~\ref{fig:RN}).
Henceforth, we will denote the extension $\phi'$ again by $\phi_{sing}$ for simplicity.

We are now ready to conclude the proof. Suppose $\phi_0$ is a solution with smooth compactly supported initial data that has finite nondegenerate energy near the whole bifurcate future Cauchy horizon. By Theorem \ref{decay.thm}, $\mathfrak L=0$. 
Then  
$$\phi=\phi_0+\beta\phi_{sing}$$
is a solution with smooth and compactly supported initial data such that $\mathfrak L\neq 0$ for every $\beta\in\mathbb R\setminus \{0\}$. In particular, using Theorem \ref{decay.thm} again, the spherically symmetric part of the solution has infinite nondegenerate energy near $\CH \cap \set{u \leq -1}$. This result can be extended to the whole ``incoming'' portion of $\CH$ in a straightforward manner; see Remark~\ref{rem:wholeCH} below. By orthogonality of the spherical modes in $L^2$ on the spheres, we thus obtain that for all $\beta\in \mathbb \bbR \setminus  \{0\}$ the solution itself also has infinite nondegenerate energy near the ``incoming'' portion of $\CH$. 

Finally, for the ``outgoing'' portion of $\CH$, it follows again from Theorem~\ref{decay.thm} that the quantity $\mathfrak L'$ defined\footnote{in an analogous manner as $\mathfrak L$} in the other exterior region $I'$ must be zero for $\phi_{0}$. Note that $\mathfrak L' = 0$ for $\phi_{sing}$ constructed above, since it vanishes identically in $I'$. Repeating the same argument as before on the region $I' \cup II$, we obtain a solution $\phi_{sing}'$ with $\mathfrak L' \neq 0$ and $\mathfrak{L} = 0$. Then $\phi = \phi_{0} + \bt (\phi_{sing} + \phi_{sing}')$ for $\bt \in \bbR \setminus \set{0}$ is a solution with infinite nondegenerate energy near the whole bifurcate Cauchy horizon, as desired.
\end{proof}

\subsubsection{Outline of the proof}\label{sec.outline}

We now describe the main steps of the proof of Theorem \ref{decay.thm}. First, we use the ideas\footnote{The analogue of $\mathfrak{L}$ was first introduced in \cite{LO1}, where we studied the sharp decay rates for the scalar field in the \emph{nonlinear} setting of dispersive solutions to the Einstein-scalar field system in spherical symmetry. In that setting, the non-vanishing of (the analogue of) $\mathfrak{L}$ implies a pointwise lower bound for the decay rate of the scalar field. } in \cite{LO1} to show that if we assume both an upper bound \eqref{R.upper.bd.ass} and $\mathfrak L\neq 0$, then we get the following lower bound:
\begin{theorem}\label{thm.decay.R}
There exists a large constant $R_1=R_1(M)>2M$ such that for any solution to \eqref{wave.eqn} with spherically symmetric initial data satisfying \eqref{data.1}-\eqref{data.3}, the following holds:
Assume that
\begin{equation}\label{R.upper.bd.ass}
\sup_{\{r= R_1\}\cap\{u\geq 1\}} u^3|\phi|\leq A'
\end{equation}
for some $A'>0$ and assume moreover that 
$$\mathfrak L\neq 0.$$
Then there exist $R=R(\mathfrak L,A',D,U_0,R_1)>R_1$ and $U=U(\mathfrak L,A',D,U_0,\rd_v(r\phi)\restriction_{C_{-U_{0}}})$ sufficiently large such that the following lower bound holds pointwise on $\gamma_R$ for $u \geq U$:
\begin{equation}\label{thm.1.lower.bd}
|\rd_v(r\phi)|(u,v)\restriction_{r=R}\geq \f{|\mathfrak L|}{8}v^{-3}
\end{equation}
\end{theorem}
In the second step, we show that if the solution $\phi$ decays sufficiently fast along the event horizon in $L^2$, then the upper bound \eqref{R.upper.bd.ass} holds (allowing us to apply Theorem \ref{thm.decay.R}), but the lower bound \eqref{thm.1.lower.bd} fails. More precisely, we have
\begin{theorem}\label{horizon.lower.bd}
Assume that $\phi$ is a solution to \eqref{wave.eqn} with spherically symmetric initial data satisfying \eqref{data.1} and \eqref{data.2}. If $\phi$ satisfies the following $L^2$ upper bound on the event horizon $\mathcal H^+$:
$$\int_{\mathcal H^+\cap \{v\geq 1\}} v^{7+\ep}(\rd_v\phi)^2 =A<\infty$$
for some $A>0$ and $\ep>0$, then for every $R\geq r_+$, we have the upper bounds
$$\sup_{\{r_+\leq r\leq R\}\cap\{v\geq 1\}} v^{3+\f \ep 2}|\phi|\leq C$$
and 
$$\int_{\{r=R,\, v\geq 1\}} v^{5}(\rd_v(r\phi))^2\leq C $$
for some $C=C(A,R,D,\ep)>0$.
\end{theorem}
In particular, this proves that under the assumption of Theorem \ref{horizon.lower.bd}, the upper bound \eqref{R.upper.bd.ass} in Theorem \ref{thm.decay.R} holds\footnote{More precisely, this is because for every fixed $R_1$, we have $u\leq C v$ for some $C=C(R_1)$.} but the lower bound \eqref{thm.1.lower.bd} fails. Thus, Theorem \ref{thm.decay.R} and Theorem \ref{horizon.lower.bd} together imply the following lower bound on the event horizon when $\mathfrak L\neq 0$:
\begin{corollary}\label{cor.horizon.lower.bd}
Let $\phi$ be a solution to \eqref{wave.eqn} with spherically symmetric initial data satisfying \eqref{data.1}-\eqref{data.3} such that
$$\mathfrak L\neq 0.$$
Then for every $\ep>0$, the following holds along the event horizon:
\begin{equation}\label{thm.2.EH}
\int_{\mathcal H^+\cap \{v\geq 1\}} v^{7+\ep}(\rd_v\phi)^2 =\infty.
\end{equation}
\end{corollary}
Finally, in the third step, we consider the region in the interior of the black hole and show that if \eqref{thm.2.EH} holds on the event horizon, then the solution has infinite nondegenerate energy near the Cauchy horizon - more precisely, \eqref{blow.up.bd} holds. We state the contrapositive as follows:
\begin{theorem}\label{final.blow.up.step}
Let $\phi$ be a solution to \eqref{wave.eqn} with spherically symmetric initial data satisfying \eqref{data.1}-\eqref{data.3}. If for some $u\leq -1$ and for some integer $\alpha_0>0$, we have
$$\int_{1}^\infty \log^{\alpha_0}(\f 1{\Omega}) (\rd_v\phi)^2(u,v)\, dv<\infty,$$
then 
$$\int_{\mathcal H^+\cap \{v\geq 1\}} \f{v^{\alpha_0}}{\log^2(1+v)}(\rd_v\phi)^2 <\infty.$$
\end{theorem}

\begin{remark} \label{rem:wholeCH}
The restriction $\set{u \leq -1}$ is purely for technical convenience, and can be dropped easily. One way is to directly use the estimates used in the proof of Theorem~\ref{final.blow.up.step} (see Section~\ref{sec.interior}), which allow us to consider other values of $u$ by a simple \emph{local} argument. Another way is to observe that both the hypothesis and the conclusion of Theorem~\ref{final.blow.up.step} are invariant under the isometry $t \mapsto t + t_{0}$, according to which $u \mapsto u - \frac{1}{2} t_{0}$ and $v \mapsto v + \frac{1}{2} t_{0}$.
\end{remark}

Combining Corollary \ref{cor.horizon.lower.bd} and Theorem \ref{final.blow.up.step} and choosing $\alp_0$ to be an integer such that $\alp_0>7$, we thus obtain Theorem \ref{decay.thm}. We remind the readers again that Theorem \ref{decay.thm} together with the construction of a solution satisfying $\mathfrak L\neq 0$ (which is achieved in Theorem \ref{construction.existence.thm}) conclude the proof of the main theorem (Theorem \ref{main.thm}, Corollary \ref{main.cor}).

\subsection{Wave equation in the exterior of the black hole}\label{sec.exterior}

We now turn to some discussions on previous works. We begin with the wave equation in the exterior region, which has been better understood and tremendous progress has been made in the past decade. The wave equation in the exterior region is relevant for our paper in two ways: First, we need as input to our main theorem some estimates that are obtained for the solutions to the wave equation in the exterior region. Second, as a consequence of the proof of the main theorem, we also show that the so-called Price's law is in a certain sense sharp along the event horizon. We state this result in Corollary \ref{linear.Price.law.sharp} below.

To summarize the known boundedness and decay results for the wave equation in the exterior of Reissner-Nordstr\"om, we introduce the notation that $\Sigma_0$ is an asymptotically flat spacelike hypersurface that penetrates the event horizon $\mathcal H^+$ and $\Sigma_\tau$ is the image of $\Sigma_0$ under the $1$-parameter family of diffeomorphisms generated by $\rd_t$. We have 
\begin{theorem}[Civin \cite{Civin}]\label{thm:Civin}
Given initial data on $\Sigma_0$ which decay towards spatial infinity and have finite nondegenerate energy, i.e., for some $D>0$,
$$E_{\Sigma_0}(\phi)\leq D,$$ 
the following bounds hold:
\begin{enumerate}
\item (Boundedness of energy) For some $C>0$, we have
$$\sup_{\tau\in [0,\infty)} E_{\Sigma_\tau}(\phi) \leq CD.$$
\item (Integrated local energy decay) For every $\delta>0$, there exists $C=C(\delta)>0$ such that

\begin{equation*}
	\int_{0}^{\infty} \bb( \int_{\Sgm_{\tau}} \frac{\chi_{PS}(r)}{r^{1+\dlt}}(\abs{n_{\Sgm_{\tau}} \phi}^{2} +\abs{\nb \phi}_{\Sgm_{\tau}}^{2}) + \frac{1}{r^{3+\dlt}} \abs{\phi}^{2} \bb) \, \ud \tau \leq C D,
\end{equation*}
where $\chi_{PS}(r)$ is a smooth cutoff which vanishes at the photon sphere $r=\f{3M+\sqrt{9M^2-8e^2}}{2}$ and $\abs{\nb \phi}_{\Sgm_{\tau}}^{2}$ is defined using the induced metric on $\Sgm_{\tau}$.
\end{enumerate}
\end{theorem}

This result is in fact a particular case of a more general theorem that holds for general subextremal Kerr-Newman spacetimes. This latter theorem is in turn based on the methods in the recent seminal work of Dafermos-Rodnianski-Shlapentokh--Rothman \cite{DRSR}, which achieved both the boundedness of energy and the integrated local energy decay estimate for the full range of subextremal Kerr spacetimes. We note that this result has its roots in the remarkable development in the past decade in understanding the decay of solutions to the linear wave equation on the exterior of black hole spacetimes. We refer the readers to \cite{AB, BSt, DRK, DRL, DRS, DRSub1, DRSub2, KW, MMTT, TT} and the references therein for a sample of such developments.

Given the result of Theorem \ref{thm:Civin} together with asymptotic flatness of the spacetime, it is known that if the initial data also have bounded higher order energies, then in fact pointwise estimates hold for the solution $\phi$. There are several approaches to such ``black box'' results, including a vector field method approach by Dafermos-Rodnianski \cite{DRNM} which has applications for nonlinear problems (see \cite{Yang}). On the other hand, the works of Tataru \cite{Ta} and Metcalfe-Tataru-Tohaneanu \cite{MTT} showed a sharper decay rate for the solution under slightly stronger assumptions on the spacetime geometry. It is easy to check that the Reissner-Nordstr\"om spacetime satisfies the assumptions required for these theorems and therefore we have the following result:
\begin{theorem}[Tataru \cite{Ta}, Metcalfe-Tataru-Tohaneanu \cite{MTT}] \label{linear.Price.law}
For sufficiently regular initial data decaying sufficient fast towards spatial infinity, the following pointwise decay estimates hold in the exterior of the black hole:
$$|\phi|\leq \begin{cases}\f{C}{(1+|v|)(1+|u|)^2} & \mbox{if } r\geq 2r_+\\ \f{C}{(1+|v|)^3} & \mbox{if } r_+\leq r< 2r_+  \end{cases},\quad |\rd_t\phi|\leq \begin{cases}\f{C}{(1+|v|)(1+|u|)^3} & \mbox{if } r\geq 2r_+\\ \f{C}{(1+|v|)^4} & \mbox{if } r_+\leq r< 2r_+  \end{cases}$$
\end{theorem}
Such decay rates are also known under the name ``Price's law'' as they were first suggested by the heuristic study of Price \cite{Price}. We remark that the rigorous proof of the Price's law decay in \emph{spherical symmetry} was achieved previously by Dafermos-Rodnianski \cite{DRPL}, who obtained slightly weaker bounds than that in Theorem \ref{linear.Price.law} but remarkably also in a \emph{nonlinear} setting. In particular, the result in \cite{DRPL} will be useful in the present work (see Theorem \ref{thm:DRPL} below).

On the other hand, one particular consequence of our approach in proving Theorem \ref{main.thm} is that we also obtain an $L^2$ lower bound for generic solutions (see Corollary \ref{cor.horizon.lower.bd}). In particular, we show that no stronger pointwise bounds than that in Theorem \ref{linear.Price.law} for $\rd_v\phi=\rd_t\phi$ on the event horizon can hold. We summarize this in the following corollary:
\begin{corollary}\label{linear.Price.law.sharp}
The decay rate for $|\rd_t\phi|$ in Theorem \ref{linear.Price.law} is sharp on the event horizon. More precisely, except for a possible co-dimension $1$ set of initial data, spherically symmetric data give rise to solutions $\phi$ to \eqref{wave.eqn} which have the property that for every $\ep>0$, there exists a sequence $v_n\to \infty$ (depending on $\ep$ and $\phi$) such that on the event horizon, we have
$$v_n^{4+\ep}|\rd_t\phi(\infty,v_n)|\to \infty.$$
\end{corollary}

\subsection{The strong cosmic censorship conjecture and previous works}\label{SCC}

As mentioned above, the presence of the smooth Cauchy horizon in the interior of the Reissner-Nordstr\"om black hole allows the maximal globally hyperbolic development to be extended smoothly but non-uniquely as solutions to the Einstein-Maxwell system. On the other hand, the celebrated strong cosmic censorship conjecture of Penrose suggests that such extensions are not possible when given generic data. More precisely, we have

\begin{conjecture}[Strong cosmic censorship]
Maximal globally hyperbolic developments for the Einstein-Maxwell system to generic asymptotically flat initial data are inextendible as suitably regular Lorentzian manifolds.
\end{conjecture}

In particular, according to the strong cosmic censorship conjecture, the smooth Cauchy horizons of Reissner-Nordstr\"om spacetimes are non-generic. Indeed, one of the early motivations for the strong cosmic censorship conjecture, in addition to the appeal of a deterministic theory, is that the Cauchy horizon appears to be linearly unstable, at least heuristically. It was already observed in the numerical work of Simpson-Penrose \cite{SP} that there is a linear instability mechanism associated to the blue shift effect along the Cauchy horizon. This led to further study of the propagation of linear test fields on a fixed Reissner-Nordstr\"om spacetime \cite{McN, CH, GSNS}. In particular, McNamara \cite{McN} showed that there exist data that can be imposed on past null infinity such that the solution is not regular at the Cauchy horizon.

On the other hand, in the early years of the conjecture, the precise implications of this linear instability was debated. In particular, it was frequently argued that the linear instability would lead to a Schwarzschild-like spacelike singularity for the nonlinear theory. The picture only became clearer after the works of Hiscock \cite{Hiscock} and Poisson-Israel \cite{PI1, PI2} which considered the coupled Einstein-Maxwell-null dust system. In particular, it was argued that general perturbations for the nonlinear system still admits a Cauchy horizon in a neighborhood of timelike infinity for which the metric remains continuous. Moreover, the ``mass inflation'' scenario was put forward, suggesting that generically the mass blows up on the Cauchy horizon and in particular the spacetime is not $C^1$ at the Cauchy horizon. This picture was finally established rigorously in the works of Dafermos \cite{D1, D2} for the Einstein-Maxwell-scalar field system in spherical symmetry.

While the above results are restricted to spherical symmetry, Dafermos-Luk \cite{DL} very recently announced the $C^0$-stability of the Kerr Cauchy horizon\footnote{The Kerr spacetime is a solution to the vacuum Einstein equations which also has a smooth Cauchy horizon. Strictly speaking, the work \cite{DL} only covers the case of Kerr, but one hopes that the methods also give an analogous result regarding perturbations of Reissner-Nordstr\"om spacetime for the Einstein-Maxwell system.}, which provided the first mathematical result regarding perturbations of the interior of the Kerr black hole \emph{without any symmetry assumptions}. More precisely, it was shown that for initial data on the event horizon which are close to and approaching the geometry of the Kerr event horizon, the solution exists all the way up to the Cauchy horizon ``in a neighborhood of timelike infinity''. The solution is moreover everywhere $C^0$-close to the Kerr solution and has a spacetime metric that is continuous up to the Cauchy horizon.

The recent work \cite{DL} in particular shows that a $C^0$ formulation of the strong cosmic censorship conjecture is false provided that the conjectural stability of the exterior region of Kerr holds true. In other words, if the exterior of Kerr is stable, then all solutions arising from data sufficiently close to Kerr spacetimes are in fact extendible with a $C^0$ metric. On the other hand, in view of the mass inflation scenario that is established in \cite{D1, D2} under additional assumptions on the flux along the event horizon, one can still hope that the conjecture holds if we require the class of ``suitably regular'' Lorentzian manifolds to be $W^{1,2}_{loc}$. As pointed out by Christodoulou \cite{Chr}, from the point of view of partial differential equations, the $W^{1,2}_{loc}$ formulation of the strong cosmic censorship conjecture has the consequence that generically the solution does not admit any extensions as weak solutions to the Einstein equations\footnote{This is because $W^{1,2}_{loc}$ is the minimal requirement for the metric to define a weak solution to the Einstein equations.}. 

Nevertheless, despite the progress in understanding the stability of the Cauchy horizons, the mechanism for which the instability occurs in $W^{1,2}_{loc}$ is not understood mathematically. In this paper, instead of discussing nonlinear problems, we return to the study of linear instability. In particular, in Theorem \ref{main.thm}, we prove that for the linear scalar wave equation, there is a global instability mechanism which generically give rise to solutions that are not $W^{1,2}_{loc}$ at the Cauchy horizon. 

In the next section, we discuss some of the known mathematical results regarding the solutions to the linear scalar wave equation - the ``poor man's'' linearized problem. In particular, a result of Dafermos shows that the solution blows up in $W^{1,2}_{loc}$ if one \emph{assumes} lower bounds regarding the global behavior of the solution on the event horizon. However, it is not known whether generic solutions obey this assumed lower bound. In contrast, in our main theorem, we proved that blow up can be guaranteed by the condition \eqref{L.condition} along null infinity, which is satisfied by solutions arising from generic Cauchy data.  We also note that many results that are known for the linear scalar wave equation on Reissner-Nordstr\"om spacetime have also been proved in the \emph{nonlinear} setting of the Einstein-Maxwell-scalar field system in \emph{spherical symmetry}. We hope that our result is also relevant in this setting (see further discussions in Section \ref{SSEMSF}).

\subsubsection{``Poor man's'' linearized problem}

One of the simplest linear problem on Reissner-Nordstr\"om spacetime is that of the linear scalar wave \eqref{wave.eqn}. It can be viewed as a ``poor man's'' version of the linearized Einstein-Maxwell system in which one suppresses the tensorial nature of the linearization as well as ignores all the lower order terms.

To further simplify the analysis for the linear stability and instability of the Cauchy horizon, one can begin with the setting where only trivial data are prescribed on the event horizon. For such data, we have both stability and instability results, which can be summarized as follows:
\begin{theorem}\label{thm.linear.hor}
Consider solutions $\phi$ to the equation \eqref{wave.eqn} with smooth initial data which vanish on the event horizon and have nondegenerate energy on $\underline C_1$ which satisfy the following bound for some $D>0$:
\begin{equation}\label{data.energy.bound}
E_{\underline C_1\cap \{u\leq U\}}(\phi)\leq De^{-\kappa_+ |U|},
\end{equation}
where $\kappa_+>0$ is given by $\kappa_+=\f{r_+-r_-}{2r_+^2}$ with $r_\pm=M\pm \sqrt{M^2-e^2}$ as before.
Then the following statements hold:
\begin{enumerate}
\item (Franzen \cite{Fra}) The solution is uniformly bounded, i.e., there exists $C>0$ such that
$$|\phi|\leq C.$$
\item (Sbierski \cite{Sbi.2}) If $\f{e^2}{M^2}>\f{4\sqrt 2}{3+2\sqrt 2}$, then the solution has finite nondegenerate energy everywhere in the interior of the black hole.
\item (Dafermos\footnote{Strictly speaking, this result is not explicitly stated in \cite{D1} but nevertheless follows from the methods in proving Theorem 2 in \cite{D1}. Moreover, notice that the statement in \cite{D1} is stated in a regular $u$ coordinate and we have translated the statement into a form using the coordinate system introduced in Section \ref{sec.geometry}.} \cite{D1}) Fix $\f{e^2}{M^2}<\f{4\sqrt 2}{3+2\sqrt 2}$. If, in addition to the bound \eqref{data.energy.bound}, there exists $U\in (-\infty,-U_0)$ and $c>0$ such that the spherically symmetric part of the initial data satisfies 
\begin{equation}\label{dafermos.trans.assumption}
\sup_{\underline C_1\cap \{u\leq U\}}|\f{\rd_u\phi}{\rd_u r}|\geq ce^{\kappa_+ s u}
\end{equation}
for\footnote{The condition $\f{e^2}{M^2}<\f{4\sqrt 2}{3+2\sqrt 2}$ guarantees that such an $s$ exists.} $0\leq s<\f 12(\f{r_+}{r_-})^2-1$, then the solution has infinite nondegenerate energy along constant $u$ null hypersurfaces intersecting the Cauchy horizon.
\end{enumerate}
\end{theorem}

The result of Dafermos can be interpreted as a blow-up statement for \emph{generic} initial data. In particular, the above theorem suggests that the Cauchy horizon is linearly unstable for a subrange of parameters $e$ and $M$ even when the data vanish on the event horizon. However, one should keep in mind that the solution on the event horizon is \emph{not} expected\footnote{This expectation indeed holds true in view of Corollary \ref{cor.horizon.lower.bd}!} to vanish for generic compactly supported data on the Cauchy hypersurface $\Sigma_0$. It turns out that there is a sense that the solution is ``more unstable'' in this case (and hence our main theorem holds). Nevertheless, we still have the following uniform boundedness result.

\begin{theorem}[Franzen \cite{Fra}]\label{thm.Fra}
Given initial data for \eqref{wave.eqn} which are smooth and compactly supported, there exists $C>0$ such that 
$$|\phi|\leq C$$
globally, including in the interior of the black hole up to the Cauchy horizon.
\end{theorem}

On the other hand, the fact that instability occurs in the full range of parameters $0<e<M$ if one takes into account the global structure of the spacetime\footnote{as opposed to only the interior of the black hole} is already suggested by the following theorem of Sbierski, which is based on a construction of Gaussian beam solutions capturing the celebrated blue shift effect of the Cauchy horizon.
\begin{figure}[h]
\begin{center}
\def\svgwidth{200px}
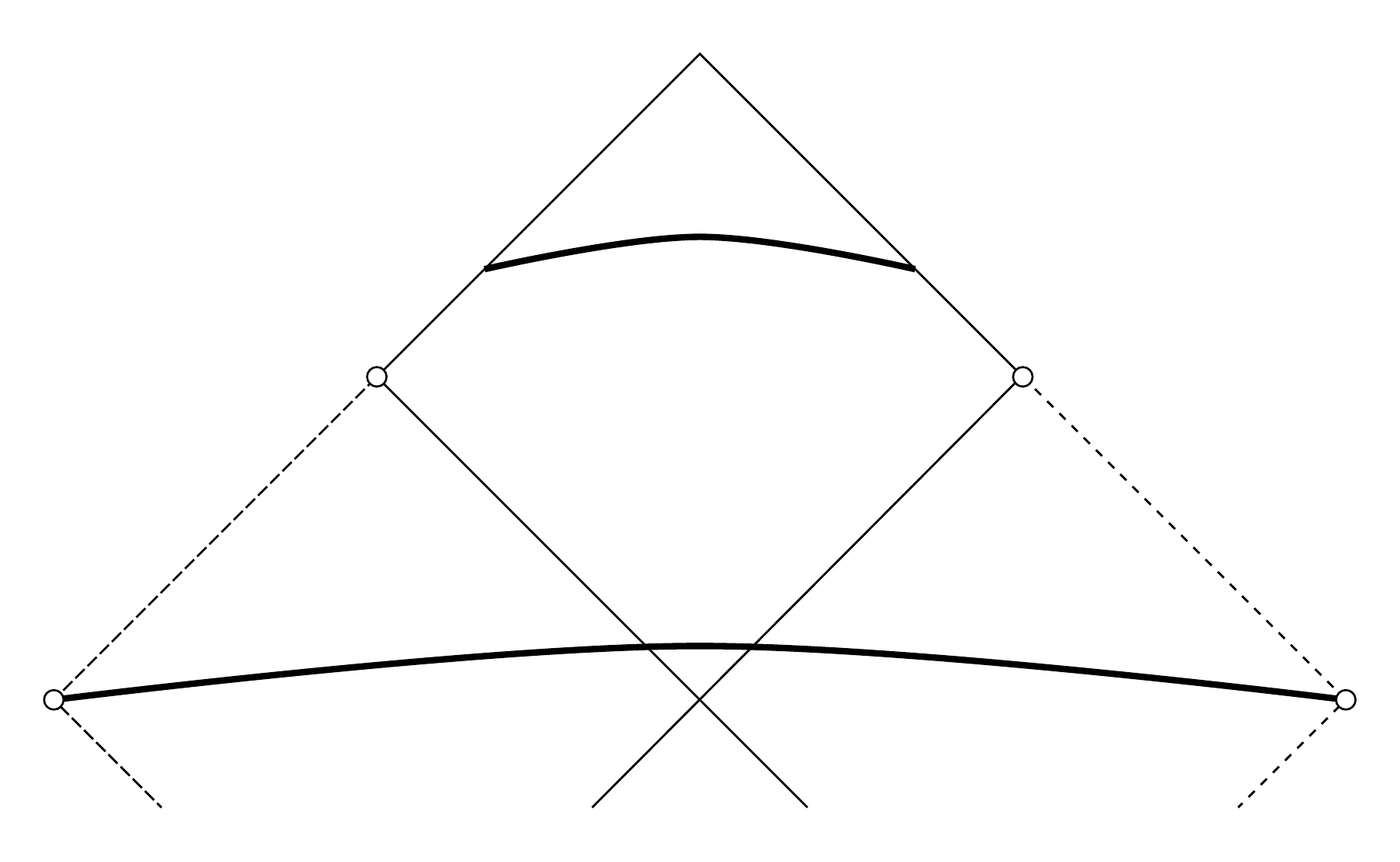 
\caption{} \label{fig:sbierski}
\end{center}
\end{figure}

\begin{theorem}[Sbierski \cite{Sbi.1}]
Let $\Sigma_0$ and $\Sigma_1$ be slices as indicated in Figure~\ref{fig:sbierski}. 
Let $E_{\Sigma_0}$ and $E_{\Sigma_1}$ be the nondegenerate energy on these two hypersurfaces as defined in Section \ref{sec.nondegenerate}.
Then there exists a sequence of solutions $\{\phi_i\}_{i\in\mathbb N}$ such that $E_{\Sigma_0}(\phi_i)=1$ but $E_{\Sigma_1}(\phi_i)\to +\infty$.
\end{theorem}

Another way to view the difference between the cases where the data are posed only in the interior of the black hole and where the data are posed on a Cauchy hypersurface is that one expects that for generic data on a Cauchy hypersurface, the solution exhibits a polynomial tail in the $v$ variable along the event horizon. Such polynomial tails as \emph{upper bound} for the solution along the event horizon has been proved rigorously (see Theorem \ref{linear.Price.law} and \cite{DRPL}) but on the other hand, Dafermos showed that if the Price's law as a lower bound is \emph{assumed} along the event horizon, then the solution $\phi$ has infinite nondegenerate energy. More precisely, we have
\begin{theorem}[Dafermos \cite{D2}]\label{Dafermos.thm}
Assume that the spherically symmetric part of the data on the event horizon satisfy the bounds
\begin{equation}\label{dafermos.horizon.assumption}
C^{-1}v^{-3p+\ep}\leq |\rd_v \phi|\leq Cv^{-p}
\end{equation}
for $v\geq V$, for some constants $C>0$, $V>0$, $p>\f 12$ and the data on a transversal null hypersurface obey
$$\sup_{\underline C_1\cap \{u\leq U\}}|\f{\rd_u\phi}{\rd_u r}|\leq D.$$
Then the solution has infinite nondegenerate energy along constant $u$ null hypersurfaces intersecting the Cauchy horizon.
\end{theorem}
\begin{remark}
In fact, given the pointwise bounds in \eqref{dafermos.horizon.assumption}, one can show more: the spherically symmetric part of $\rd_v\phi$ obeys a pointwise lower bound
$$|\rd_v\phi|\geq c v^{-3p+\ep}$$
in the interior of the black hole and therefore is not in $W^{1,p}_{loc}$ for all $p>1$.
\end{remark}
Theorem \ref{Dafermos.thm} in particular reduces the problem of generic blow up to showing that solutions to generic initial data on $\Sigma_0$ obey the bounds \eqref{dafermos.horizon.assumption}. However, showing this bound for generic data remains an open problem and in fact it is not known whether there exist\footnote{One na\"{i}ve way to try to construct Cauchy data such that the solutions satisfy the desired lower bounds is to impose data on the event horizon and future null infinity and solve backwards. However, in view of the red-shift effect on the event horizon, one expects that the Cauchy data constructed in this manner are generically non-regular. See the discussions in \cite{DRSRS}.} \emph{any} regular Cauchy data on $\Sigma_0$ that give rise to a solution verifying \eqref{dafermos.horizon.assumption}. On the other hand, we show in the present paper that for generic data, a weaker lower bound than \eqref{dafermos.horizon.assumption} holds. Namely, we prove that for generic data, \eqref{thm.2.EH} in Corollary \ref{cor.horizon.lower.bd} is verified along the event horizon and this can be viewed as a polynomial lower bound in an $L^2$-averaged sense. Moreover, we show that this is already sufficient to guarantee that the solution is in fact not in $W^{1,2}_{loc}$ near the Cauchy horizon, thus resolving the problem of $W^{1,2}_{loc}$ blow up for generic Cauchy data.

\subsubsection{Spherically symmetric Einstein-Maxwell-scalar field system}\label{SSEMSF}

It turns out that many of the insights gained in the ``poor man's'' linearized problem can in fact be applied to the \emph{nonlinear} problem for the Einstein-Maxwell-scalar field system with \emph{spherically symmetric} data in perturbative regimes. 
More precisely, in a series of works\footnote{We also refer the readers to \cite{KomThe, CGNS2, CGNS3} for some related recent results in this direction.} \cite{D1,D2,D3}, Dafermos studied the following system with spherically symmetric initial data:
\begin{equation*}
\begin{split}
Ric_{\mu\nu}-\f12 g_{\mu\nu} R=&2(T^{(sf)}_{\mu\nu}+T^{(em)}_{\mu\nu}),\\
T^{(sf)}_{\mu\nu}=&\rd_\mu\phi\rd_\nu\phi-\f 12 g_{\mu\nu} (g^{-1})^{\alp\beta}\rd_\alp\phi\rd_{\beta}\phi,\\
T^{(em)}_{\mu\nu}=&(g^{-1})^{\alp\bt}F_{\mu\alp}F_{\nu\bt}-\f 14 g_{\mu\nu}(g^{-1})^{\alp\bt}(g^{-1})^{\gamma\sigma}F_{\alp\gamma}F_{\bt\sigma},
\end{split}
\end{equation*}
where $\phi$ is a scalar function and $F$ is a $2$-form satisfying
$$\Box_g\phi=0\,\quad dF=0\,\quad F^{\mu\nu}{ }{ }_{;\nu}=0.$$
Here, the subscript $(\cdot)_{;\nu}$ stands for a covariant derivative. For this system with spherically symmetric data, it is known that small data give rise to globally $C^0$ bounded solutions. This can be viewed as a version\footnote{In fact, much more is proved in \cite{D2, DRPL}: Given (potentially large) data with non-vanishing charge such that the solution approaches a subextremal Reissner-Nordstr\"om spacetime in the exterior region, there exists ``a piece of Cauchy horizon'' near future timelike infinity such that the metric and the scalar field remains continuous up to the Cauchy horizon.} of Theorem \ref{thm.Fra} in the (much more complicated) nonlinear setting:
\begin{theorem}[Dafermos \cite{D2,D3}, Dafermos-Rodnianski \cite{DRPL}]
Given $2$-ended asymptotically flat spherically symmetric initial data globally close to Reissner-Nordstr\"om data, the maximal globally hyperbolic development has the same causal structure as that of Reissner-Nordstr\"om and is globally $C^0$-close to Reissner-Nordstr\"om. In particular, the metric and the scalar field is continuous up to the global bifurcate Cauchy horizon.
\end{theorem}

Moreover, in the settings analogous to that in Theorems \ref{thm.linear.hor} and \ref{Dafermos.thm}, Dafermos also showed that the Cauchy horizon is indeed singular in the sense that the Hawking mass blows up. 
\begin{theorem}
Let $(\mathcal M, g, F,\phi)$ be a spherically symmetric solution to the Einstein-Maxwell-scalar field equation.
\begin{enumerate}
\item (Dafermos \cite{D1}) Consider initial data which is exactly that of Reissner-Nordstr\"om on the event horizon and assume that \eqref{dafermos.trans.assumption} holds. Then the Hawking mass is identically infinite along the Cauchy horizon.
\item (Dafermos \cite{D3}) If the scalar $\phi$ satisfies \eqref{dafermos.horizon.assumption} along the event horizon and the data are regular on $\underline C_1$, then the Hawking mass blows up identically along the Cauchy horizon.
\end{enumerate}
Furthermore, in both of these scenarios, the scalar field is not in $W^{1,2}_{loc}$ in any neighborhood of any point on the Cauchy horizon.
\end{theorem}

On the other hand, as in the case for the linear wave equation, the blow-up mechanism when data are posed on a Cauchy hypersurface is much less understood. In particular, it is not known whether there exists a single solution with data that are sufficiently regular and close to that of Reissner-Nordstr\"om such that the solution is singular in a neighborhood of the Cauchy horizon. Nevertheless, in view of Theorem \ref{main.thm} for the linear wave equation, we make the following conjecture:
\begin{conjecture}
Generic smooth $2$-ended asymptotically flat spherically symmetric initial data that are globally close to Reissner-Nordstr\"om data give rise to solutions that are not in $W^{1,2}_{loc}$ near the Cauchy horizon.
\end{conjecture}
\subsection{Outline of the paper} We conclude our introduction with an outline of the remainder of the paper. Our proof of Theorem \ref{decay.thm} follows the outline sketched in Section \ref{sec.outline}. In Section \ref{sec.lower.bd.R}, we prove Theorem \ref{thm.decay.R}; in Section \ref{sec.horizon}, we prove Theorem \ref{horizon.lower.bd}; and in Section \ref{sec.interior}, we prove Theorem \ref{final.blow.up.step}. Finally, we conclude the paper by proving Theorem \ref{construction.existence.thm} in Section \ref{sec.construction}.
\\
\\
\noindent
{\bf Acknowledgments:} The authors thank Mihalis Dafermos and Jan Sbierski for stimulating discussions. They also thank Mihalis Dafermos for very helpful comments on a preliminary version of the manuscript. Part of this work was carried out in Princeton University, MIT and Cambridge University. The authors thank these institutions for their hospitality. J. Luk is supported by the NSF Postdoctoral Fellowship DMS-1204493. S.-J. Oh is a Miller Research Fellow, and thanks the Miller Institute at UC Berkeley for support.

\section{Lower bound on $\{r=R\}$}\label{sec.lower.bd.R}
Our main goal in this section is to prove Theorem \ref{thm.decay.R}, i.e., we show that if we have the upper bound \eqref{R.upper.bd.ass} and moreover $\mathfrak{L}\neq 0$, then there exists $R$ sufficiently large such that along the $\{r=R\}$ curve, we have the lower bound 
\begin{equation} \label{eq:lb:r=R}
|\rd_v(r\phi)|(u,v) \restriction_{\{r=R\}}\geq \f{|\mathfrak L|} 8 v^{-3}.
\end{equation}

We first recall\footnote{Strictly speaking, in \cite{DRPL} Theorems~9.4 and ~9.5 is stated for the full nonlinear Einstein-Maxwell-(real)-scalar field system. However, its linear analogue on a Reissner-Nordstr\"om background can be easily inferred by following the outline provided in Section 11 of the same paper. Moreover, in Theorems~9.4 and ~9.5 of \cite{DRPL}, the precise dependence of the constant in \eqref{eq:DRPL:conc} on the size of the initial data is not stated. Nevertheless, since the equation is linear, we can apply the closed graph theorem to immediately deduce Theorem~\ref{thm:DRPL}.} the following estimate from \cite{DRPL}:
\begin{theorem}[Dafermos-Rodnianski \cite{DRPL}, Theorems~9.4, ~9.5 and Section~11] \label{thm:DRPL}
Let $\phi$ be a spherically symmetric $C^{1}$ real-valued solution to the linear wave equation on the exterior region of a Reissner-Nordstr\"om space-time to the future of $\uC_{1} \cup C_{-U_{0}}$, where $U_{0} > 0$ is a fixed number. Suppose that, for real numbers $1 < \omg \leq 3$ and $B > 0$, the data on $C_{-U_{0}} \cap \set{v \geq 1}$ obey
\begin{equation} \label{eq:DRPL:hyp:C}
	\sup_{C_{-U_{0}} \cap \set{v \geq 1}} r^{2} \abs{\frac{\rd_{v} \phi}{\rd_{v} r}} \leq B, \quad
	\sup_{C_{-U_{0}} \cap \set{v \geq 1}} r^{\omg} \abs{\frac{\rd_{v}(r \phi)}{\rd_{v} r}} \leq B.
\end{equation}
Suppose furthermore that the data on $\uC_{1} \cap \set{u \geq -U_{0}}$ obeys
\begin{equation} \label{eq:DRPL:hyp:uC}
	\sup_{\uC_{1} \cap \set{u \geq -U_{0}}} \abs{\frac{\rd_{u} \phi}{\rd_{u} r}} \leq B, \quad
	\sup_{\uC_{1} \cap \set{u \geq -U_{0}}} \abs{\frac{\rd_{u}(r \phi)}{\rd_{u} r}} \leq B.
\end{equation}
Then there exists a constant $C = C(\omg, U_{0}, R_0) > 0$ such that in the region $r \geq R_0$, we have
\begin{equation} \label{eq:DRPL:conc}
	(1+u_{+})^{\omg-1} \abs{r \phi} (u, v) \leq C B,
\end{equation}
where $u_{+} := \max \set{u, 0}$. 

\end{theorem}

In order to obtain a lower bound on $\{r=R\}$, we first need an improved decay estimate for $\phi$ compared to Theorem \ref{thm:DRPL} in the region $\{r\geq R_1\}$ for some sufficiently large $R_1$. To prove this improved decay estimate, we rely crucially on the \emph{assumption} that $|\phi|$ has an improved decay on a constant $r$-curve, i.e., $\{r=R_1\}$. The proof uses the sharp decay estimate of Dafermos-Rodnianski in Theorem \ref{thm:DRPL} and also applies some ideas from \cite{LO1}. More precisely, we have
\begin{proposition}\label{imp.decay}
Suppose that $\phi$ is a solution to \eqref{wave.eqn} with initial data satisfying \eqref{data.1}-\eqref{data.3}.
There exists $R_1=R_1(M)$ sufficiently large such that if
$$\sup_{\set{r= R_1} \cap \set{u \geq 1}} u^3|\phi|\leq A',$$
for some $A'>0$, then the following estimate holds for some $C=C(A',D,U_0,R_1)$:
$$ \sup_{\set{r\geq R_1} \cap \set{u \geq 1}} u^3|\phi|\leq C. $$
\end{proposition}
\begin{proof}
Let $\mathcal B(U)=\sup_{u\in [1,U], r\in [R_1,\infty)} u^3|\phi|$ for $U > 1$.
We recall our convention that for every $u$, $v_{R_1^*}(u)$ denotes the unique $v$ value such that $r(u,v_{R_1^*}(u))=R_1$. Given a point $(u,v)$ such that $r\geq R_1$, we integrate along a constant $u$ curve from $(u,v_{R_1^*}(u))$ to $(u,v)$ to get
\begin{equation*}
\begin{split}
u^3|\phi(u,v)| \leq &u^3|\phi(u,v_{R_1^*}(u))|+\f{u^3}{r(u,v)}\int_{v_{R_1^*}(u)}^{v} |\rd_v(r\phi)|(u,v') dv'\\
\leq &A'+Cu^3\sup_{v'\in [v,\infty)} |\rd_v(r\phi)|(u,v')
\end{split}
\end{equation*}
since $\f{v-v_{R_1^*}(u)}{r(u,v)}\leq C$ for some $C=C(R_1)$. This implies that
\begin{equation}\label{B.trivial}
\mathcal B(U)\leq A'+C\sup_{r\geq R_1,\,u\leq U,\,v'\in [v,\infty)} u^3|\rd_v(r\phi)|(u,v').
\end{equation}
We now use the wave equation
$$\rd_u\rd_v (r\phi)=-\f{2(M-\f{e^2}{r})\Omg^2\phi}{r^2}$$
and integrate along a constant $v$ curve starting from the initial data to the point $(u,v)$. We will assume that $u\geq 1$. In the case where $u\geq \f v2$, we further divide the integral into the regions $-U_0 \leq u'\leq \f v2$ and $\f v2\leq u'\leq u$. More precisely, we have
\begin{align}
&|\rd_v(r\phi)(u,v)| \label{imp.decay.1} \\
\leq &|\rd_v(r\phi)(-U_0,v)|+\int_{-U_0}^u \f{2(M-\f{e^2}{r})\Omg^2|\phi|}{r^2}(u',v)\,du' \notag\\
\leq &|\rd_v(r\phi)({-U_0},v)|+\int_{-U_0}^{\min\{\f v2,u\}} \f{2(M-\f{e^2}{r})\Omg^2r|\phi|}{r^3}(u',v)\,du'+\int_{\f v2}^{\max\{\f v2,u\}} \f{2(M-\f{e^2}{r})\Omg^2|\phi|}{r^2}(u',v)\,du' \notag \\
\leq &|\rd_v(r\phi)({-U_0},v)|+\int_{-U_0}^{\f v2} \f{C}{(1+u'_{+})^2} \times\f{2(M-\f{e^2}{r})\Omg^2}{r^3}(u',v)\,du'+\int_{\f v2}^{\max\{\f v2,u\}} \f{2(M-\f{e^2}{r})\Omg^2|\phi|}{r^2}(u',v)\,du' \notag \\
\leq &|\rd_v(r\phi)({-U_0},v)|+Cu^{-3}+Cu^{-3}\mathcal B(u) \int_{\f v2}^{\max\{\f v2,u\}} \f{2(M-\f{e^2}{r})\Omg^2}{r^2}(u',v)\,du' \notag
\end{align}
for some $C=C(U_0)$. In the second to the last line above, we used Theorem \ref{thm:DRPL} to control the second term.
In the last line above, we have used the following two facts. First, we have the estimate
$$\int_{-U_0}^{\f v2} \f{1}{(1+u'_{+})^2} \times\f{2(M-\f{e^2}{r})\Omg^2}{r^3}(u',v)\,du'\leq \f{2M}{r^3(\f v2,v)} \int_{-U_0}^{\infty}\f{du'}{(1 + u'_{+})^2} \leq Cu^{-3}$$
since $u\leq Cv\leq Cr(\f v2,v)$ for $R_1$ sufficiently large. Second, for $\f v2\leq u'\leq \max\{\f v2,u\}$, we have $\sup_{\f v2\leq u'\leq \max\{\f v2,u\}}u'^{-3}\leq C u^{-3}$.

To proceed, we use the assumption \eqref{data.1} on the initial data to obtain for some $C=C(U_0)$ that
\begin{equation}\label{imp.decay.2}
u^3|\rd_v(r\phi)({-U_0},v)|\leq Cr^3|\rd_v(r\phi)({-U_0},v)|\leq CD.
\end{equation}
Combining \eqref{imp.decay.1} and \eqref{imp.decay.2}, we obtain for some $C=C(D,U_0)$ that
\begin{equation*}
\begin{split}
u^3|\rd_v(r\phi)(u,v)|\leq &C+\mathcal B(u) \int_{\f v2}^{\max\{\f v2,u\}} \f{2(M-\f{e^2}{r})\Omg^2}{r^2}(u',v)\,du'.
\end{split}
\end{equation*}
Combining this with \eqref{B.trivial}, we get for some $C=C(A',D,U_0)$ that
$$\mathcal B(u)\leq C+\mathcal B(u) \int_{\f v2}^{\max\{\f v2,u\}} \f{2(M-\f{e^2}{r})\Omg^2}{r^2}(u',v)\,du'.$$
Finally, notice that by choosing $R_1$ sufficiently large, 
$$\int_{\f v2}^{\max\{\f v2,u\}} \f{2(M-\f{e^2}{r})\Omg^2}{r^2}(u',v)\,du'$$
can be made arbitrarily small. Therefore, we have
$$\sup_u\mathcal B(u)\leq C$$
for some $C=C(A',D,U_0,R_1)$ and the conclusion follows.
\end{proof}

{\bf From this point onward, we take $R_1$ to be fixed such that the conclusion of Proposition \ref{imp.decay} applies.} Equipped with the upper bound for $|\phi|$ provided by Proposition \ref{imp.decay}, we can now turn to the main goal of this section, i.e., to obtain a lower bound for $\rd_v(r\phi)$. We will achieve this in two steps. First, in Proposition \ref{decay.far} immediately below, we show the desired lower bound along a curve $\{u=\eta v\}$ where $\eta$ will be chosen to be a sufficiently small constant. We will then show in Proposition \ref{lower.bd.R} that this lower bound can be propagated in the whole region of $r\geq R$, as long as $R$ is chosen to be sufficiently large. 

Before we proceed, we make one simplifying assumption that $\mathfrak L>0$. The case $\mathfrak L<0$ will be dealt with at the end of this section using the symmetry $\phi\mapsto -\phi$ for the wave equation. We begin with the lower bound on $\{u=\eta v\}$:
\begin{proposition}\label{decay.far}
Suppose that $\phi$ is a solution to \eqref{wave.eqn} with initial data satisfying \eqref{data.1}-\eqref{data.3}. Moreover assume that 
$$\mathfrak L> 0.$$
Then there exists $\eta>0$ sufficiently small (depending on $D$ and $\mathfrak L$) and $U$ sufficiently large such that the following lower bound holds:
$$\rd_v(r\phi)(u,v)\geq  \f{\mathfrak L}{4} v^{-3}$$
whenever $U\leq u\leq \eta v$.
\end{proposition}
\begin{proof}
Using the wave equation \eqref{re.wave.eqn.ext}, we have
\begin{equation}\label{decay.far.0}
v^3\rd_v(r\phi)(u,v)=v^3\rd_v(r\phi)({-U_0},v)-v^3\int_{-U_0}^u \f{2(M-\f{e^2}{r})\Omg^2}{r^3} \, r \phi (u',v)\,du'.
\end{equation}
First, by the decay of $u^2r|\phi|$ given by Theorem \ref{thm:DRPL}, the following holds for $U_1>{-U_0}$ such that $U_1$ is sufficiently large (depending on $D$ and $\mathfrak L$):
\begin{equation}\label{decay.far.est.1}
\sup_{u\geq U_1,\, r(u,v)\geq R_0}\int_{U_1}^{u} 2(M-\f{e^2}{r})\Omg^2 \, r |\phi| (u',v)\,du'\leq \f{\mathfrak L}{8}.
\end{equation}
In particular, the limiting value of this integral at future null infinity (i.e., $v=\infty$) is also bounded as follows:
\begin{equation}\label{decay.far.est.2}
\int_{U_1}^\infty 2M |\Phi|(u')  \,du' \leq \f{\mathfrak L}{8}.
\end{equation}
On the other hand, since $\lim_{v\to \infty} r\phi(u',v)\to \Phi(u')$ pointwise, we have by the dominated convergence theorem that 
\begin{equation}\label{decay.far.est.3}
|\int_{-U_0}^{U_1} 2(M-\f{e^2}{r})\Omg^2 r \phi(u',v)\,du'-\int_{-U_0}^{U_1} 2 M \Phi (u')  \,du'|\leq \f{\mathfrak L}{8}
\end{equation}
for $v$ sufficiently large.
Also, for $v$ sufficiently large, in the region $\{u\leq \eta v\}$, we have
\begin{equation}\label{vr.comp}
|\f{v^3}{r^3}-1|\leq C\eta.
\end{equation}
Moreover, by choosing $V_1$ to be sufficiently large and requiring $v\geq V_1$, we have simultaneously $U\leq \eta v$ and $r(U,v)\geq R$. We can therefore apply \eqref{vr.comp} to get
\begin{equation*}
\begin{split}
&|v^3\int_{-U_0}^{U_1} \f{2(M-\f{e^2}{r})\Omg^2\phi}{r^2}(u',v)\,du'-\int_{-U_0}^{U_1} 2(M-\f{e^2}{r})\Omg^2 r \phi(u',v)\,du'|\\
\leq &C\eta |\int_{-U_0}^{U_1} 2(M-\f{e^2}{r})\Omg^2 r \phi(u',v)\,du'|\\
\leq &C\eta,
\end{split}
\end{equation*}
where the last line holds because the integral $|\int_{-U_0}^{U_1} 2(M-\f{e^2}{r})\Omg^2 r \phi(u',v)\,du'|$ is bounded using Theorem \ref{thm:DRPL}. Therefore, we can choose $\eta$ to be sufficiently small (depending on $D$ and $\mathfrak L$) such that 
\begin{equation}\label{decay.far.est.4}
\begin{split}
|v^3\int_{-U_0}^{U_1} \f{2(M-\f{e^2}{r})\Omg^2\phi}{r^2}(u',v)\,du'-\int_{-U_0}^{U_1} 2(M-\f{e^2}{r})\Omg^2 r \phi(u',v)\,du'|\leq \f{\mathfrak L}{8}.
\end{split}
\end{equation}
Combining the above bounds we know that for $u\geq U_1$ and $v\geq V_1$, we have
\begin{equation}\label{decay.far.1}
\begin{split}
&|v^3\int_{-U_0}^u \f{2(M-\f{e^2}{r})\Omg^2}{r^2} \phi (u',v)\,du'-\int_1^\infty 2M \Phi(u')  \,du'|\\
\leq &|\int_{-U_0}^{U_1} 2(M-\f{e^2}{r})\Omg^2 r \phi(u',v)\,du'-\int_{-U_0}^{U_1} 2M \Phi(u')  \,du'|\\
&+|v^3\int_{-U_0}^{U_1} \f{2(M-\f{e^2}{r})\Omg^2\phi}{r^2}(u',v)\,du'-\int_{-U_0}^{U_1} 2(M-\f{e^2}{r})\Omg^2 r \phi(u',v)\,du'|\\
&+|v^3\int_{U_1}^u \f{2(M-\f{e^2}{r})\Omg^2}{r^2} \phi (u',v)\,du|+|\int_{U_1}^\infty 2M \Phi(u')  \,du'|\\
\leq &\f{\mathfrak L}{2},
\end{split}
\end{equation}
where we have used \eqref{decay.far.est.3}, \eqref{decay.far.est.4}, \eqref{decay.far.est.1} and \eqref{decay.far.est.2} for the first, second, third and fourth term respectively.

On the other hand, by \eqref{data.3}, we can choose $V_2$ sufficiently large such that
\begin{equation}\label{decay.far.2}
\sup_{v\geq V_2}|v^3\rd_v(r\phi)({-U_0},v)-\lim_{v'\to \infty}r^3\rd_v(r\phi)({-U_0},v')|\leq \f{\mathfrak L}4.
\end{equation}
Therefore, combining \eqref{decay.far.1} and \eqref{decay.far.2}, we obtain the following estimate for $u\geq U_1$, $v\geq \max\{V_1,V_2\}$ and $u\leq \eta v$:
\begin{equation}\label{decay.far.3}
\begin{split}
&|v^3\rd_v(r\phi)({-U_0},v)-v^3\int_{-U_0}^u \f{2(M-\f{e^2}{r})\Omg^2}{r^3} \, r \phi (u',v)\,du'|\\
\geq &|\lim_{v\to \infty}r^3\rd_v(r\phi)({-U_0},v)-\int_{-U_0}^\infty 2M \Phi(u')  \,du'|-\f{\mathfrak L}{2}-\f{\mathfrak L}{4}\\
= &\mathfrak L-\f{3 \mathfrak L}{4}= \f{\mathfrak L}{4}.
\end{split}
\end{equation}
Now, we can choose $U>U_1$ to be sufficiently large such that $\{u\geq U\}\cap\{u\leq \eta v\}\subset\{v\geq \max\{V_1,V_2\}\}$. Finally, returning to the equation \eqref{decay.far.0} and using the bound \eqref{decay.far.3} we get that
$$v^3 \rd_v(r\phi)(u,v)\geq \f{\mathfrak L}{4},$$
for every $(u,v)$ such that $u\in [U,\eta v]$, as desired.
\end{proof}

Our final step in this section is to show that under the assumption \eqref{R.upper.bd.ass}, we have a lower bound for $|\rd_v(r\phi)|$ on $\gamma_R:=\{r=R\}$ as long as $R$ is sufficiently large. This thus concludes the proof of Theorem \ref{thm.decay.R}. To obtain this conditional lower bound, we combine conditional improved decay estimate in Proposition \ref{imp.decay} with the lower bound derived in Proposition \ref{decay.far}.
\begin{proposition}\label{lower.bd.R}
Suppose that $\phi$ is a solution to \eqref{wave.eqn} with initial data satisfying \eqref{data.1}-\eqref{data.3}. 
Assume that
$$\sup_{r= R_1} u^3|\phi|\leq A'$$
for some $A'>0$, where $R_1$ is such that the conclusion of Proposition \ref{imp.decay} holds. 
Moreover assume that 
$$\mathfrak L> 0.$$
There exists $R>R_1$ sufficiently large (depending on $\mathfrak L$, $A'$, $D$, $U_0$ and $R_1$) such that the following lower bound holds:
$$\rd_v(r\phi)(u,v)\restriction_{r=R}\geq \f{\mathfrak L}{8}v^{-3}$$
for $u\geq U$, where $U$ is as in Proposition \ref{decay.far}.
\end{proposition}
\begin{proof}
Take $(u,v)$ such that $r(u,v)= R$ and $u\geq U$, where $U$ is as in Proposition \ref{decay.far}. We again use the wave equation \eqref{re.wave.eqn.ext} and integrate from the point $(\eta v, v)$ to $(u,v)$ to get
\begin{equation*}
\begin{split}
|\rd_v(r\phi)(u,v)-\rd_v(r\phi)(\eta v,v)|
\leq &\int_{\eta v}^u \f{2(M-\f{e^2}{r})\Omg^2|\phi|}{r^2}(u',v)\,du'\\
 \leq & C \eta^{-3} v^{-3} \int_{\eta v}^u \f{2(M-\f{e^2}{r})\Omg^2}{r^2}(u',v)\,du'
\end{split}
\end{equation*}
for some $C=C(A',D,U_0,R_1)$ by Proposition \ref{imp.decay}. Now notice that for $(u,v)$ such that $r(u,v)=R$, we have
$$\int_{\eta v}^u \f{2(M-\f{e^2}{r})\Omg^2}{r^2}(u',v)\,du'\leq \f{CM}{R}\to 0 \quad\mbox{as }R\to \infty.$$
In particular, $R$ can be chosen to be sufficiently large (depending on $\mathfrak L$, $A'$, $D$, $U_0$, $R_1$ and $\eta$) such that
\begin{equation}\label{finite.R.diff.bd}
|\rd_v(r\phi)(u,v)-\rd_v(r\phi)(\eta v,v)|\leq \f{\mathfrak L}{8} v^{-3}.
\end{equation}
Finally, recall the lower bound in Proposition \ref{decay.far}
$$\rd_v(r\phi)(\eta v,v)\geq \f{\mathfrak L}{4} v^{-3}$$ 
for $u\geq U$. Combining this with \eqref{finite.R.diff.bd} gives the desired lower bound on $\{r=R\}$.
\end{proof}

We have thus proved Theorem \ref{thm.decay.R} for $\mathfrak L>0$. In the case of $\mathfrak L<0$, notice that we can apply the above result to show that 
$$-\rd_v(r\phi)(u,v)\restriction_{r=R}\geq -\f{\mathfrak L}{8}v^{-3}$$
for $u\geq U$. Therefore, Theorem \ref{thm.decay.R} follows.

\section{Lower bound on the event horizon}\label{sec.horizon}
In this section, the analysis continues to take place in the exterior of the black hole. As mentioned in the introduction, our main goal is to prove Theorem \ref{horizon.lower.bd}.

Recalling the setting in Theorem \ref{horizon.lower.bd}, we assume that there exists $\eps > 0$ and $A>0$ such that
\begin{equation}\label{contra}
\int_{\EH \cap \set{v \geq 1}} v^{7 + \eps} (\rd_{v} \phi)^{2} = A < \infty.
\end{equation}
We remark that in the coordinates $(u,v)$, \eqref{contra} reads
\begin{equation*} \tag{\ref{contra}$'$}
	\int_{1}^{\infty} v^{7+\eps} (\rd_{v} \phi)^{2}(\infty, v) \, \ud v = A < \infty.
\end{equation*}
We begin with a simple consequence of \eqref{contra}, namely a pointwise decay estimate for $\phi$ on $\mathcal H^+$.
\begin{proposition}\label{prop:EH:phi}
Let $\phi$ be a solution to the linear wave equation with spherically symmetric data verifying \eqref{data.1}-\eqref{data.2} and moreover obeying \eqref{contra}. Then on the event horizon $\EH = \set{(u, v) : u = \infty}$, we have the bound
\begin{equation} \label{eq:EH:phi}
|\phi|\leq Cv^{-3-\frac{\ep}{2}},
\end{equation}
where $C = C(A) > 0$.
\end{proposition}

\begin{proof}
According to the results in \cite{DRPL}, $\phi\to 0$ along the event horizon as $v\to \infty$. Therefore, by a direct application of the Cauchy-Schwarz inequality, we have
\begin{equation*}
\begin{split}
|\phi|(\infty,v)&\leq \int_v^{\infty}|\rd_v\phi| dv \\
&\leq \big(\int_v^{\infty} (v')^{-7-\ep} dv'\big)^{\frac 12} \big(\int_v^{\infty} (v')^{7+\ep}(\rd_v\phi)^2(\infty,v') dv'\big)^{\frac 12}\leq C A v^{-3-\frac{\ep}{2}}. \qedhere
\end{split}
\end{equation*}
\end{proof}

Next, we utilize the red-shift effect of the event horizon $\EH = \set{r = r_{+}}$ (see Dafermos-Rodnianski \cite{DRPL}) to propagate decay into a neighborhood $\set{r \leq R_{2}}$ of $\EH$.
%

\begin{proposition}\label{prop:red-shift}
Let $\phi$ be a solution to the linear wave equation with spherically symmetric data verifying \eqref{data.1}-\eqref{data.2} and moreover obeying \eqref{contra}. Then for $R_{2} > r_{+}$ sufficiently close, we have
\begin{align} 
\sup_{(u, v) \in \set{r \leq R_{2}} \cap \set{v \geq 1}} v^{\frac{7}{2}+\frac{\eps}{2}}\abs{\frac{\rd_{u} \phi}{\Omg^{2}}}(u,v)
	\leq& C,  \label{eq:red-shift:1} \\
\int_{C_{u} \cap \set{r \leq R_{2}} \cap \set{v \geq 1}} v^{7+\eps} (\rd_{v} \phi)^{2} \leq& C, \label{eq:red-shift:2}
\end{align}
as well as
\begin{align} 
	\sup_{(u, v) \in \set{r \leq R_{2}}} v^{3+\frac{\eps}{2}} \abs{\phi} \leq & C, \label{eq:red-shift:r=R2:phi} \\
	\sup_{r_+\leq r\leq R_2}\int_{\gmm_{r} \cap \set{v \geq 1}} v^{7+\eps} (\rd_{v} \phi)^{2}  \leq & C, \label{eq:red-shift:r=R2:dvphi}
\end{align}
for a constant $C = C(A, \eps, R_{2},D) > 0$. Here the notation $\gmm_{r}$ denotes the constant $r$-curve parametrized by the spacetime coordinate $v$, and the integration over $\gmm_{r}$ is with respect to $\ud v$.
\end{proposition}

\begin{proof} 
From the wave equation \eqref{wave.eqn.ext} we can derive
\begin{align*}
	\rd_{v} (\frac{\rd_{u} \phi}{\Omg^{2}}) 
	=& - \bb( \frac{\rd_{v} \Omg^{2}}{\Omg^{2}} + \frac{\Omg^{2}}{r} \bb) \frac{\rd_{u} \phi}{\Omg^{2}} + \frac{1}{r} \rd_{v} \phi.
\end{align*}

Hence, multiplying by $(v + V)^{\alp} \, \Omg^{-2} \rd_{u} \phi$, where $V > 0$ is a large constant to be chosen below, we arrive at the identity
\begin{equation} \label{eq:extr:red-shift}
	\frac{1}{2} \rd_{v} \bb( (v+V)^{\alp} \big( \frac{\rd_{u} \phi}{\Omg^{2}} \big)^{2} \bb) + \bb( \frac{\rd_{v} \Omg^{2}}{\Omg^{2}} + \frac{\Omg^{2}}{r} - \frac{\alp}{2(v+V)} \bb) (v+V)^{\alp} \big( \frac{\rd_{u} \phi}{\Omg^{2}} \big)^{2} = (v+V)^{\alp} \rd_{v} \phi \frac{\rd_{u} \phi}{\Omg^{2}}.
\end{equation}

A simple computation shows that $\Omg^{-2} \rd_{v} \Omg^{2}$ is equal to a \emph{positive} constant on $\EH$; this computation captures the red-shift effect along $\EH$. Hence, for $R'_{2} > r_{+}$ sufficiently close to $r_{+}$, there exists $c > 0$ such that
\begin{equation*}
	\frac{\rd_{v} \Omg^{2}}{\Omg^{2}}= \frac{2M}{r^{2}} - \frac{2 e^{2}}{r^{3}} \geq c \quad \hbox{ for } r_{+} \leq r \leq R'_{2}.
\end{equation*}
Moreover, there exists $V = V(R_{2}', \alp, c) > 0$ such that 
\begin{equation*}
	\frac{\rd_{v} \Omg^{2}}{\Omg^{2}} - \frac{\alp}{v+V}  \geq \frac{c}{2} \quad \hbox{ for } r_{+} \leq r \leq R'_{2} \hbox{ and } v \geq 1.
\end{equation*} 
Integrating \eqref{eq:extr:red-shift} along a curve $\set{u} \times [v_{1}, v_{2}]$ such that $\sup_{v \in [v_{1}, v_{2}]} r(u,v) \leq R'_{2}$, and using the bound $v \leq v + V \leq (V+1) v$ for $v \geq 1$ to absorb $V$ into the constant $C$ (to simplify the notation), we obtain
\begin{equation} \label{eq:red-shift:duphi}
	\sup_{v \in [v_{1}, v_{2}]} v^{\alp} (\frac{\rd_{u} \phi}{\Omg^{2}})^{2}(u, v) + \int_{v_{1}}^{v_{2}} v^{\alp} (\frac{\rd_{u} \phi}{\Omg^{2}})^{2} (u, v)\, \ud v 
	\leq  C \bb( v_{1}^{\alp} (\frac{\rd_{u} \phi}{\Omg^{2}})^{2}(u, v_{1}) + \int_{v_{1}}^{v_{2}} v^{\alp} (\rd_{v} \phi)^{2}(u, v) \, \ud v \bb).
\end{equation}
where $C = C(R_{2}', \alp, c) > 0$. 

On the other hand, from the wave equation we have
\begin{align*}
	\rd_{u} (r \rd_{v} \phi) 
	= - \rd_v r \rd_{u} \phi
\end{align*}
Hence multiplying by $v^{\alp} r \rd_{v} \phi$ and integrating over $(u,v) \in [u_{1}, u_{2}] \times [v_{1}, v_{2}]$, we have
\begin{equation} \label{eq:red-shift:rdvphi}
	\int_{v_{1}}^{v_{2}} r^{2}  v^{\alp} (\rd_{v} \phi)^{2}(u_{1}, v) \, \ud v
	=  \int_{v_{1}}^{v_{2}} r^{2} v^{\alp} (\rd_{v} \phi)^{2}(u_{2}, v) \, \ud v 
		+ \int_{v_{1}}^{v_{2}} \int_{u_{1}}^{u_{2}} v^{\alp} (\rd_v r \rd_{u} \phi) (r \rd_{v} \phi) \, \ud u \ud v.
\end{equation}
To treat the last term, we apply the Cauchy-Schwarz inequality and write
\begin{align*}
\int_{v_{1}}^{v_{2}} \int_{u_{1}}^{u_{2}} v^{\alp} \abs{(\rd_{v} r \rd_{u} \phi) (r \rd_{v} \phi)} \, \ud u \ud v	
\leq &
\bb(\sup_{(u, v) \in [u_{1}, u_{2}] \times [v_{1}, v_{2}]} (\rd_v r)^{\f 12} v^{\frac{\alp}{2}} \abs{\frac{\rd_{u} \phi}{\Omg^{2}}} \bb)  \bb( \int_{u_{1}}^{u_{2}} \sup_{v \in [v_{1}, v_{2}]} \Omg^{2}  \, \ud u \bb)  \\
& \times  \sup_{u \in [u_{1}, u_{2}]} 
	\bb( \int_{v_{1}}^{v_{2}} \rd_v r \, \ud v \bb)^{\frac{1}{2}} 
	\bb( \int_{v_{1}}^{v_{2}} r^{2} v^{\alp} (\rd_{v} \phi)^{2} \, \ud v \bb)^{\frac{1}{2}}
\end{align*}
It can be easily verified that $0 \leq \rd_v r \leq 1$ in the exterior of the black hole. Moreover, as $\rd_u r = -\Omg^{2}$ and $\Omg^{- 2} \rd_{v} \Omg^{2}  > 0$ for $r > r_{+}$ (i.e., $\Omg^2$ is increasing in $v$), it follows that
\begin{equation*}
	\int_{u_{1}}^{u_{2}} \sup_{v \in [v_{1}, v_{2}]} \Omg^{2}  \, \ud u 
	= \int_{u_{1}}^{u_{2}} - \rd_u r (u, v_{2}) \, \ud u
	= r(u_{1}, v_{2}) - r(u_{2}, v_{2}).
\end{equation*}
Finally, we use the trivial bound $\int_{v_{1}}^{v_{2}} \rd_v r \ud v \leq 2 \sup_{v \in [v_{1}, v_{2}]} r(u, v)$, since we do not expect this quantity to be small in rectangles with long $v$-length.
Therefore, the last term in \eqref{eq:red-shift:rdvphi} is bounded by
\begin{equation} \label{eq:red-shift:rdvphi:error}
\begin{aligned}
\leq & \bb(\sup_{(u, v) \in [u_{1}, u_{2}] \times [v_{1}, v_{2}]} \abs{v^{\frac{\alp}{2}} \frac{\rd_{u} \phi}{\Omg^{2}}} \bb)  (r(u_{1}, v_{2}) - r(u_{2}, v_{2})) \\
	& \times \bb(2 \sup_{(u, v) \in [u_{1}, u_{2}] \times [v_{1}, v_{2}]} r(u, v) \bb)^{\frac{1}{2}} \sup_{u \in [u_{1}, u_{2}]} \bb( \int_{v_{1}}^{v_{2}} r^{2} v^{\alp} (\rd_{v} \phi)^{2} \, \ud v \bb)^{\frac{1}{2}}.
\end{aligned}
\end{equation}
Let $R_{2}\in [r_+,R_2']$ be a constant to be determined below, and consider a rectangle $[u_{1}, u_{2}] \times [v_{1}, v_{2}] \subseteq \set{r_{+} \leq r \leq R_{2}}$. Observe that $r(u_{1}, v_{2}) - r(u_{2}, v_{2}) \leq R_{2} - r_{+}$; moreover, $r$ is bounded from below and above by $r_{+}$ and $R_{2} \leq R_{2}'$, respectively. Combining \eqref{eq:red-shift:rdvphi} and \eqref{eq:red-shift:rdvphi:error}, we arrive at
\begin{equation} \label{eq:red-shift:dvphi}
	\sup_{u \in [u_{1}, u_{2}]} \int_{v_{1}}^{v_{2}} v^{\alp} (\rd_{v} \phi)^{2} \, \ud v
	\leq C \bb( \int_{v_{1}}^{v_{2}} v^{\alp} (\rd_{v} \phi)^{2}(u_{2}, v) \, \ud v + (R_{2} - r_{+})^{2} \sup_{(u, v) \in [u_{1}, u_{2}] \times [v_{1}, v_{2}]} v^{\alp} \big(\frac{\rd_{u} \phi}{\Omg^{2}}\big)^{2} \bb)
\end{equation}
where $C = C(R_{2}', r_{+}) > 0$. Taking $R_{2} - r_{+}$ sufficiently small (depending on $R_{2}', r_{+}, \alp$ and $c$), \eqref{eq:red-shift:duphi} and \eqref{eq:red-shift:dvphi} imply
\begin{equation} \label{eq:red-shift:combine}
\begin{aligned}
& \hskip-2em
	\sup_{(u, v) \in [u_{1}, u_{2}] \times [v_{1}, v_{2}]} v^{\alp} (\frac{\rd_{u} \phi}{\Omg^{2}})^{2}
	+ \sup_{u \in [u_{1}, u_{2}] }\int_{v_{1}}^{v_{2}} v^{\alp} (\frac{\rd_{u} \phi}{\Omg^{2}})^{2} (u, v)\, \ud v 
	+ \sup_{u \in [u_{1}, u_{2}]} \int_{v_{1}}^{v_{2}} v^{\alp} (\rd_{v} \phi)^{2} \, \ud v \\
\leq & C \bb( \sup_{u \in [u_{1}, u_{2}]} v_{1}^{\alp} \big(\frac{\rd_{u} \phi}{\Omg^{2}} \big)^{2}(u, v_{1}) + \int_{v_{1}}^{v_{2}} v^{\alp} (\rd_{v} \phi)^{2} (u_{2}, v) \, \ud v \bb).
\end{aligned}
\end{equation}
Applying \eqref{eq:red-shift:combine} with $\alp = 7 + \eps$ to rectangles of the form $[u, \infty) \times [1, v]$, where $(u, v) \in \set{r \leq R_{2}}$ and $v \geq 1$, we obtain \eqref{eq:red-shift:1} and \eqref{eq:red-shift:2}. Moreover, \eqref{eq:red-shift:r=R2:phi} can be proved by integrating \eqref{eq:red-shift:1} over $[u, \infty) \times \set{v}$, using \eqref{eq:EH:phi} and the fact that $\Omg^{2} = - \rd_{u} r$. Finally, repeating the above argument to the region $\set{r \leq R_{2}} \cap \set{v \geq 1}$, \eqref{eq:red-shift:r=R2:dvphi} follows; we omit the details. \qedhere
\end{proof}

Notice that if $R\leq R_2$, then we have already obtained Theorem \ref{horizon.lower.bd}. We therefore assume that $R>R_2$ and we show that the pointwise and integrated decay estimates \eqref{eq:red-shift:r=R2:phi} and \eqref{eq:red-shift:r=R2:dvphi} can be propagated to the curve $\gmm_{R} = \set{r = R}$ (with some loss in the exponent).
\begin{proposition} \label{prop:r=R}
Let $R > R_{2}$, where $R_{2}$ is as in Proposition~\ref{prop:red-shift}. Let $\phi$ be a solution to the linear wave equation with spherically symmetric data verifying \eqref{data.1}-\eqref{data.2} and moreover obeying \eqref{contra}. Then in the region $\set{r \leq R}$, we have 
\begin{align} 
	\sup_{\set{r \leq R} \cap \set{v \geq 1}} v^{3 + \frac{\eps}{2}} \abs{\phi} 
	\leq& C, \label{eq:r=R:phi} \\
	\int_{\gmm_{R} \cap \set{v \geq 1}} \frac{v^{6+\eps}}{\log^{2}(1+v)} (\rd_{v} \phi)^{2}
	 \leq& C, \label{eq:r=R:dvphi}
\end{align}
for a constant $C = C(A, \eps, R, R_{2}) > 0$.
\end{proposition}

\begin{proof} 

The idea is to use the (1+1)-dimensional energy estimate in the space-like direction using the multiplier $(\rd_{v} - \rd_{u}) \phi$ and establish \eqref{eq:r=R:dvphi} on dyadic segments $\gmm_{R} \cap \set{v_{0} \leq v \leq 2 v_{0}}$ for $v_{0} \in 2^{\bbN}$. We will also obtain estimates for the $L^2$ norm of $\rd_u\phi$ along $\underline C_v$ with appropriate $v$-decay, from which the pointwise bound \eqref{eq:r=R:phi} will follow.

In order to proceed, we need to introduce some notation. By a slight abuse of notation, let $\gmm_{r^{\ast}}$ denote the constant $r^{\ast}$ curve, which may be parametrized by $v$ or $u$ as
\begin{equation*}
	\gmm_{r^{\ast}} = \set{(u_{r^{\ast}}(v), v)} = \set{(u, v_{r^{\ast}}(u))}
\end{equation*}
where $u_{r^{\ast}}(v) = v - r^{\ast}$ and $v_{r^{\ast}}(u) = u + r^{\ast}$. 
As usual, we integrate functions over $\gmm_{r^{\ast}}$ using the $v$-parametrization. Note that the integral of the same function performed using the $u$-parametrization gives the same value.
We also define the segments
\begin{equation*}
	\uC_{v}(r^{\ast}_{1}, r^{\ast}_{2}) := \set{(u, v) : u_{r^{\ast}_{1}}(v) \leq u \leq u_{r^{\ast}_{2}}(v)}, \quad
	C_{u}(r^{\ast}_{1}, r^{\ast}_{2}) := \set{(u, v) : v_{r^{\ast}_{1}}(u) \leq v \leq v_{r^{\ast}_{2}}(u)}.
\end{equation*}
Given $v_{0} \geq 1$, denote by $\gmm^{(v_{0})}_{r^{\ast}}$ the segment
\begin{equation*}
	\gmm^{(v_{0})}_{r^{\ast}} := \gmm_{r^{\ast}} \cap \set{v_{r^{\ast}}(u_{R^{\ast}}(v_{0})) \leq v \leq 2 v_{0}}
\end{equation*}
where $R^{\ast}$ is the $r^{\ast}$ value at $r = R$. Note that $\gmm^{(v_{0})}_{R^{\ast}} = \gmm_{R} \cap \set{v_{0} \leq v \leq 2v_{0}}$, which is the segment we wish to estimate at the end. Finally, let $\calD^{(v_{0})}(r^{\ast}_{1}, r^{\ast}_{2})$ denote the region bounded by $\gmm^{(v_{0})}_{r^{\ast}_{1}}$, $\gmm^{(v_{0})}_{r^{\ast}_{2}}$, $\uC_{2v_{0}}(r^{\ast}_{1}, r^{\ast}_{2})$ and $C_{u_{R^{\ast}}(v_{0})}(r^{\ast}_{1}, r^{\ast}_{2})$. A Penrose diagram representation of these objects is given in Figure~\ref{fig:no-shift}.

\begin{figure}[h]
\begin{center}
\def\svgwidth{220px}
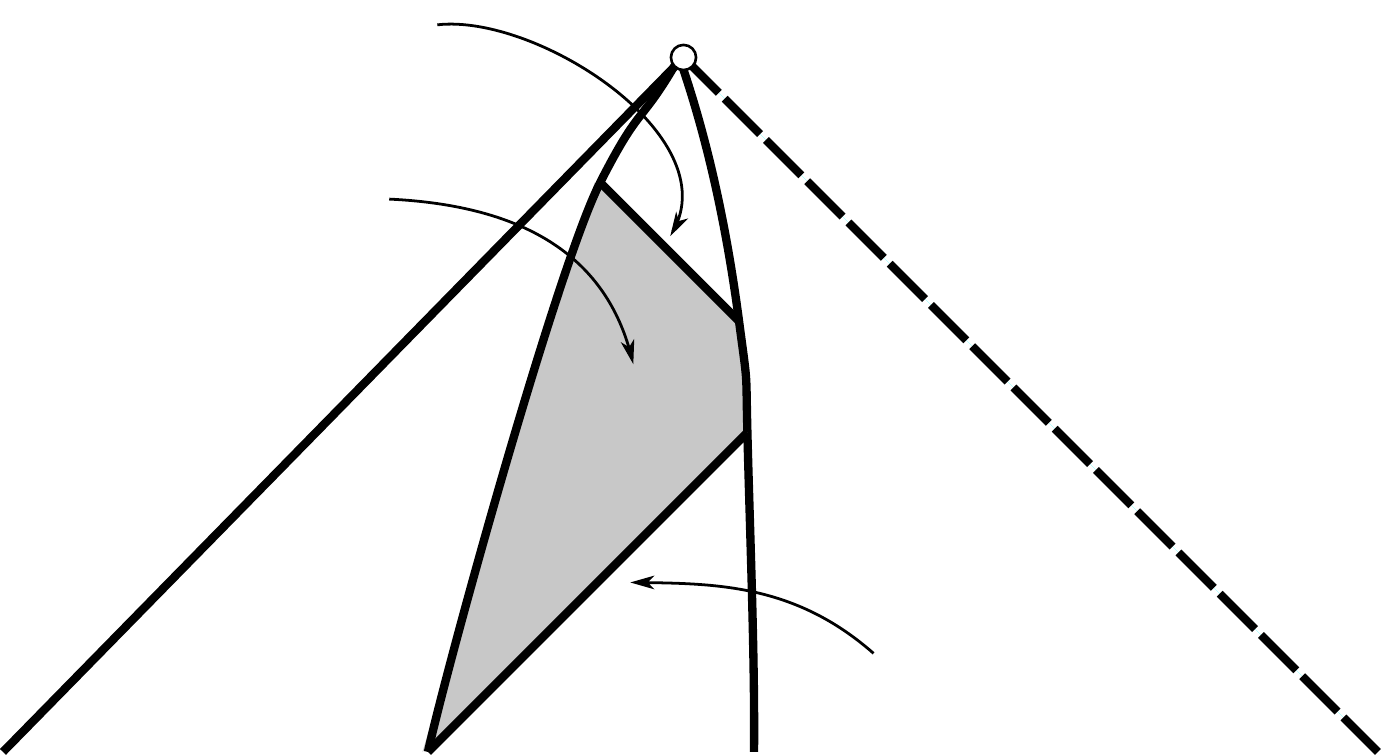 
\caption{} \label{fig:no-shift}
\end{center}
\end{figure}

For every $v_{0} \geq 1$, we claim that the following inequality holds:
\begin{equation} \label{eq:r=R:energy}
	\sup_{r^{\ast} \in [R_{2}^{\ast}, R^{\ast}]} \bb( \int_{\gmm^{(v_{0})}_{r^{\ast}}} (\rd_{v} \phi)^{2} + (\rd_{u} \phi)^{2} \bb)
	+ \sup_{v \in [v_{0}, 2v_{0}]} \int_{\uC_{v}(R^{\ast}_{2}, R^{\ast})} (\rd_{u} \phi)^{2} 
	\leq C \bb( \int_{\gmm^{(v_{0})}_{R^{\ast}_{2}}} (\rd_{v} \phi)^{2} + (\rd_{u} \phi)^{2} \bb)
\end{equation}
where $R^{\ast}_{2}$ is the $r^{\ast}$ value at $r = R_{2}$ and $C = C(R^{\ast}, R_{2}^{\ast}) > 0$. 
To prove \eqref{eq:r=R:energy}, we multiply the wave equation for $\phi$ by $(\rd_{v} - \rd_{u}) \phi$, which gives
\begin{equation} \label{phi.dr.ed}
	-\frac{1}{2} \rd_{v} (\rd_{u} \phi)^{2} + \frac{1}{2} \rd_{u} (\rd_{v} \phi)^{2} = (- \frac{\Omg^{2}}{r} \rd_{u} \phi + \frac{\Omg^{2}}{r} \rd_{v} \phi) (\rd_{v} - \rd_{u}) \phi.
\end{equation}
We then integrate \eqref{phi.dr.ed} by parts on regions of the form $\calD^{(v_{0})}(R^{\ast}_{2}, r^{\ast}) \cap \set{v \leq v_{1}}$ for $r^{\ast} \in [R^{\ast}_{2}, R^{\ast}]$ and $v_{0} \leq v_{1} \leq 2v_{0}$ with respect to the measure $du\,dv$. Then using the fact that $r$ is localized in $R_{2} \leq r \leq R$, we have $\frac{\Omg^{2}}{r} = \frac{1}{r}(1 - \frac{2M}{r} + \frac{e^{2}}{r^{2}})$, $\rd_{v} r$ and $|\rd_{u} r|$ are uniformly bounded from the above and below (independently of $v_{0}$). Therefore, we have
\begin{equation*} 
\begin{split}
	&\sup_{r^{\ast} \in [R_{2}^{\ast}, R^{\ast}]} \bb( \int_{\gmm^{(v_{0})}_{r^{\ast}}} (\rd_{v} \phi)^{2} + (\rd_{u} \phi)^{2} \bb)
	+ \sup_{v \in [v_{0}, 2v_{0}]} \int_{\uC_{v}(R^{\ast}_{2}, R^{\ast})} (\rd_{u} \phi)^{2} \\
	\leq &C \bb( \int_{\gmm^{(v_{0})}_{R^{\ast}_{2}}} (\rd_{v} \phi)^{2} + (\rd_{u} \phi)^{2}+\int_{R_2^{\ast}}^{R^{\ast}} (\int_{\gmm^{(v_{0})}_{r^{\ast}}} (\rd_v\phi)^2+(\rd_u\phi)^2) dr^{\ast} \bb)
\end{split}
\end{equation*}
for some $C=C(R_2^{\ast},R^{\ast})>0$. From this, we easily obtain \eqref{eq:r=R:energy} using Gronwall's inequality.

We are now ready to conclude the proof of \eqref{eq:r=R:dvphi}. Note that 
\begin{equation*}
	\gmm^{(v_{0})}_{R^{\ast}_{2}} = \gmm_{R^{\ast}_{2}} \cap \set{ v_{0} - R^{\ast} + R^{\ast}_{2} \leq v \leq 2v_{0}} 
	\subseteq \gmm_{R^{\ast}_{2}} \cap \set{ \frac{1}{2}v_{0} \leq v \leq 2v_{0}} 
\end{equation*}
for $v_{0}$ sufficiently large, i.e., $v_{0} \geq 2 (R^{\ast} - R^{\ast}_{2})$. 
Note furthermore that, combining \eqref{eq:red-shift:1} and \eqref{eq:red-shift:r=R2:dvphi} in Proposition~\ref{prop:red-shift}, we have
\begin{equation} \label{eq:r=R:dyadic-R2}
	\int_{\gmm_{R^{\ast}_{2}} \cap \set{\frac{1}{2} v_{0} \leq v \leq 2v_{0}}} v^{6+\eps} \big( (\rd_{v} \phi)^{2} + (\rd_{u} \phi)^{2} \big) \leq C
\end{equation}
for some $C = C(A, \eps, R_{2}, D) > 0$ independent of $v_{0}$. Summing up the dyadic bounds obtained from \eqref{eq:r=R:energy} and \eqref{eq:r=R:dyadic-R2} (for large $v_{0} \in 2^{\bbN}$), we obtain \eqref{eq:r=R:dvphi}. We remark that the loss $\log^{-2}(1+v)$ allows us to gain summability. On the other hand, again by \eqref{eq:r=R:energy} (in particular, the last term on the left-hand side) and \eqref{eq:r=R:dyadic-R2}, we also have the bound
\begin{equation*}
	\sup_{v \geq 1} v^{6+\eps} \int_{\uC_{v}(R^{\ast}_{2}, R^{\ast})} (\rd_{u} \phi)^{2} \leq C.
\end{equation*}
At this point, \eqref{eq:r=R:phi} follows from the preceding bound and \eqref{eq:red-shift:r=R2:phi} in Proposition \ref{prop:red-shift}. \qedhere

%
\end{proof}

We are now ready to conclude the proof of Theorem~\ref{horizon.lower.bd}.
\begin{proof} [Proof of Theorem~\ref{horizon.lower.bd}]
Let $(u_{R^{\ast}}(v), v)$ be a parametrization of $\gmm_{R}$, i.e., $u_{R^{\ast}}(v) = v - R^{\ast}$. Note, moreover, that $\rd_{v} r = \Omg^{2}$ is constant on $\gmm_{R}$; we denote this value by $\Omg_{R}^{2}$. Combining \eqref{eq:r=R:phi} and \eqref{eq:r=R:dvphi}, we see that
 \begin{equation*}
	\int_{\gmm_{R} \cap \set{v \geq 1}} v^{5} (\rd_{v} (r \phi))^{2} 
	\leq 2 R^{2} \int_{1}^{\infty} v^{5} (\rd_{v} \phi)^{2} (u_{R^{\ast}}(v), v) \, \ud v
		+ 2 \Omg^{2}_{R} \int_{1}^{\infty} v^{5} \phi^{2}(u_{R^{\ast}}(v), v) \, \ud v 
	\leq C.
\end{equation*}
\end{proof}

\section{Interior of the black hole}\label{sec.interior}
We now turn to the final part of the proof of Theorem \ref{decay.thm}, i.e., we prove Theorem \ref{final.blow.up.step}. We work in the interior of the black hole. The main idea is that since we are in spherical symmetry, we can solve the wave equation in the spacelike direction toward future timelike infinity $i^+$ and control the ``initial data term'' using the assumption of Theorem \ref{decay.thm}. In this context, the \emph{blue shift} effect which is at the first place the source of the instability is seen in the analysis as a \emph{red shift} effect, and we can therefore perform an iteration to prove the decay of $\phi$ along $\mathcal H^+$ (see related analysis in \cite{DRS, DRdS}).

Before we proceed, we set up some notation for this section. Let
\begin{equation*}
	\Gmm_{\tau} := \Gmm^{(1)}_{\tau} \cup \Gmm^{(2)}_{\tau}
\end{equation*}
where
\begin{equation*}
	\Gmm^{(1)}_{\tau} = 	\set{(-\tau, v) : v \geq \tau}, \quad
	\Gmm^{(2)}_{\tau} = \set{(u, \tau) : u \leq \tau}.
\end{equation*}
We also define
\begin{equation*}
	\CH(\tau_{1}, \tau_{2}) = \CH \cap \set{\tau_{1} \leq u \leq \tau_{2}}, \quad
	\EH(\tau_{1}, \tau_{2}) = \EH \cap \set{\tau_{1} \leq v \leq \tau_{2}}.
\end{equation*}
Denote by $\calD(\tau_{1}, \tau_{2})$ the region bounded by $\Gmm^{(1)}_{\tau_{1}}$, $\Gmm^{(2)}_{\tau_{1}}$, $\CH(\tau_{1}, \tau_{2})$, $\EH(\tau_{1}, \tau_{2})$, $\Gmm^{(1)}_{\tau_{2}}$, $\Gmm^{(2)}_{\tau_{2}}$. A Penrose diagram representation of these objects is provided in Figure~\ref{fig:interior}. Note that we will be integrating on the sets $\calD(\tau_{1}, \tau_{2})$, $\Gmm^{(1)}_{\tau_{1}}$, $\Gmm^{(2)}_{\tau_{1}}$, $\CH(\tau_{1}, \tau_{2})$, $\EH(\tau_{1}, \tau_{2})$, $\Gmm^{(1)}_{\tau_{2}}$, $\Gmm^{(2)}_{\tau_{2}}$ and we will use the convention for the volume elements introduced in Section \ref{sec.notation}.

\begin{figure}[h]
\begin{center}
\def\svgwidth{220px}
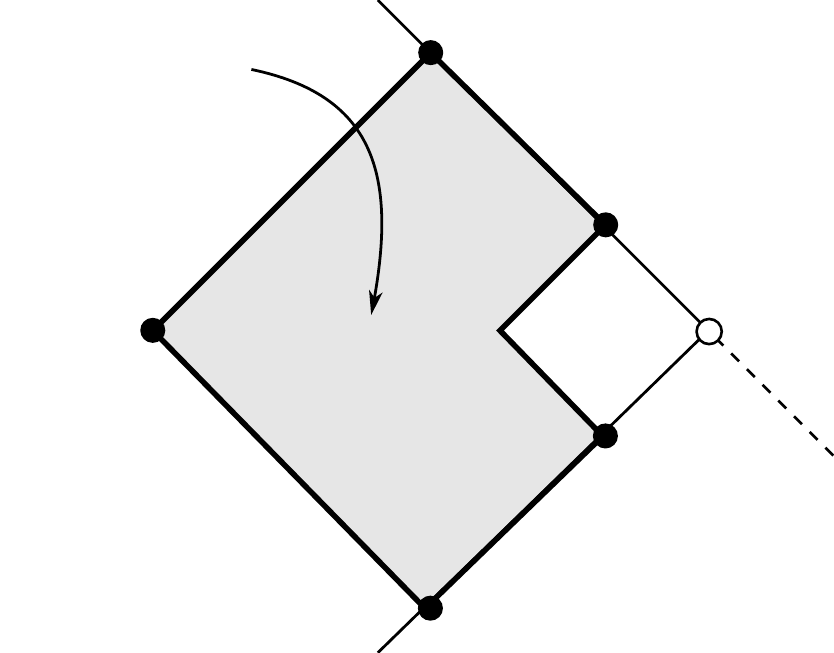 
\caption{} \label{fig:interior}
\end{center}
\end{figure}



The assumption in Theorem \ref{final.blow.up.step} implies that there exists $\tau \geq 1$ and some integer $\alp_0>7$ such that 
\begin{equation} \label{eq:contra.2.pre}
	\int_{\Gmm^{(1)}_{\tau}} v^{\alp_{0}} (\rd_{v} \phi)^{2} < \infty.
\end{equation}
At this point, it turns out to be convenient to make a few reductions to bring \eqref{eq:contra.2.pre} to a form that is easier to use.
First, for simplicity, we henceforth restrict ourselves to the case $\tau = 1$; the reader may check that the argument below remains the same when \eqref{eq:contra.2.pre} holds for other values of $\tau$. Next, on $\Gmm^{(1)}_{1}$, note that
\begin{equation*}
	\lim_{v \to \infty} \frac{v}{\log \frac{1}{\Omg} (-1, v)} = \frac{1}{\abs{\kpp_{-}}},
\end{equation*}
where $\kpp_{-} = \frac{r_{-} - r_{+}}{2 r_{-}^{2}}$ is a strictly negative number. Hence \eqref{eq:contra.2.pre} for $\tau = 1$ is equivalent to
\begin{equation} \label{contra.2}
	\int_{\Gmm^{(1)}_{1}} (1 + \chi_{1}(r) \log^{\alp_{0}} (\frac{1}{\Omg})) (\rd_{v} \phi)^{2}  =: A < \infty,
\end{equation}
where $\chi_{1}(r)$ is a smooth cutoff near $\CH$ to be fixed below (see \eqref{chi1.def}). 

Our goal now is to show that \eqref{contra.2} together with the assumptions \eqref{data.1}-\eqref{data.2} on the initial data lead to the bound
\begin{equation}\label{contra.2.goal}
	\int_{\EH\cap\{v\geq 1\}} \f{v^{\alp_0}}{\log^2(1+v)} (\rd_{v} \phi)^{2} < \infty,
\end{equation}
along $\mathcal H^+$ towards $i^{+}$.

In Propositions \ref{prop:intr:energy}-\ref{prop:intr:red-shift:final}, we will derive some identities and estimates in general for spherically symmetric solutions to the wave equation in the interior of the black hole. Then, in Proposition \ref{prop:intr:decay}, we will use the assumption \eqref{contra.2} to derive \eqref{contra.2.goal}. We begin with an energy inequality in the direction of $i^{+}$.
\begin{proposition} \label{prop:intr:energy}
For every $\tau_{1}, \tau_{2}$ such that $1 \leq \tau_{1} \leq \tau_{2}$, we have
\begin{equation} \label{eq:intr:energy}
\begin{aligned}
&\hskip-2em
	\int_{\Gmm_{\tau_{2}}^{(1)}} r^{2} (\rd_{v} \phi)^{2} + \int_{\Gmm_{\tau_{2}}^{(2)}} r^{2} (\rd_{u} \phi)^{2} 
	+ \int_{\CH(\tau_{1}, \tau_{2})} r^{2} (\rd_{u} \phi)^{2} + \int_{\EH} r^{2} (\rd_{v} \phi)^{2} \\
	\leq & C \bb( \int_{\Gmm^{(1)}_{\tau_{1}}} r^{2} (\rd_{v} \phi)^{2} + \int_{\Gmm^{(2)}_{\tau_{2}}} r^{2} (\rd_{u} \phi)^{2} \bb)
\end{aligned}
\end{equation}
for some universal $C>0$.
\end{proposition}
\begin{proof} 
The proof follows by simply integrating by parts over $\calD(\tau_{1}, \tau_{2})$ the identity
\begin{equation*}
	0 = r^{2} (\rd_{u} \phi - \rd_{v} \phi) \bb( \rd_{u} \rd_{v} \phi - \frac{\Omg^{2}}{r} \rd_{u} \phi - \frac{\Omg^{2}}{r} \rd_{v} \phi \bb),
\end{equation*} 
using the fact that $\rd_{v} r = \rd_{u} r = - \Omg^{2}$ in our coordinates. Note that the result is simply the energy identity corresponding to the stationary Killing vector field, but for propagation in the space-like direction. \qedhere
\end{proof}

Using the energy inequality, we now establish an integrated local energy decay estimate.
\begin{proposition} \label{prop:intr:ILED}
For every $\tau_{1}, \tau_{2}$ such that $1 \leq \tau_{1} \leq \tau_{2}$, we have
\begin{equation} \label{eq:intr:ILED}
	\int_{\calD(\tau_{1}, \tau_{2})} \Omg^{2} \bb( (\rd_{u} \phi)^{2} + (\rd_{v} \phi)^{2} \bb)
	\leq C \bb( \int_{\Gmm^{(1)}_{\tau_{1}}} (\rd_{v} \phi)^{2} + \int_{\Gmm^{(2)}_{\tau_{2}}} (\rd_{u} \phi)^{2} \bb)
\end{equation}
for some universal $C>0$.
\end{proposition}
\begin{proof} 
Given $N \leq 0$ to be chosen, consider
\begin{align*}
	0 =& \iint_{\calD(\tau_{1}, \tau_{2})}	
			r^{N} (\rd_{u} \phi + \rd_{v} \phi) \bb( \rd_{u} \rd_{v} \phi - \frac{\Omg^{2}}{r} \rd_{u} \phi - \frac{\Omg^{2}}{r} \rd_{v} \phi \bb) \, \ud u \ud v \\
	=& \iint - \frac{N}{2} (\rd_{v} r) r^{N-1} (\rd_{u} \phi)^{2} - \frac{N}{2} (\rd_{u} r) r^{N-1} (\rd_{v} \phi)^{2}\\
		& - \iint r^{N-1} \Omg^{2} \bb( (\rd_{u} \phi)^{2} + (\rd_{v} \phi)^{2} + 2 \rd_{u} \phi \rd_{v} \phi \bb) + (\hbox{boundary terms})\\
	=& \iint r^{N-1} \Omg^{2} \bb( (\frac{N}{2} - 1 ) (\rd_{u} \phi)^{2} + ( \frac{N}{2} - 1 ) (\rd_{v} \phi)^{2} - 2 \rd_{u} \phi \rd_{v} \phi \bb)
		+ (\hbox{boundary terms}).
\end{align*}
Choosing $N = -2$, we see that the space-time integral on the last line equals
\begin{equation*}
	- \iint r^{N-1} \Omg^{2} \bb( (\rd_{u} \phi + \rd_{v} \phi)^{2} + (\rd_{u} \phi)^{2} + (\rd_{v} \phi)^{2} \bb)
\end{equation*}
which is non-negative. Controlling the boundary terms by Proposition~\ref{prop:intr:energy}, which is possible since $r$ is bounded from the above and below in the interior region, we obtain \eqref{eq:intr:ILED}.
\end{proof}

Our next proposition captures the \emph{red-shift} effect along the Cauchy horizon as we approach $i^{+}$. We introduce a smooth cutoff $\chi_{1}(r)$ such that
\begin{equation}\label{chi1.def}
	\chi_{1}(r) = 
	\left\{
\begin{array}{cl}	
	1 & \hbox{ for } r_{-} \leq r \leq r^{(1)} \\
	0 & \hbox{ for } r \geq r^{(1)} + (r^{(1)} - r_{-}),
\end{array}. \right.
\end{equation}
where $r^{(1)} > r_{-}$ is to be specified below.

\begin{proposition} \label{prop:intr:red-shift}
For every $\alp \geq 0$, there exists a constant $C = C(\alp) > 0$ such that for $1 \leq \tau_{1} \leq \tau_{2}$, we have
\begin{equation} \label{eq:intr:red-shift}
\begin{aligned}
& \hskip-2em
	\int_{\Gmm^{(1)}_{\Gmm_{\tau_{2}}}} \chi_{1}(r) \log^{\alp} (\frac{1}{\Omg}) (\rd_{v} \phi)^{2}
	+ \alp \int_{\calD(\tau_{1}, \tau_{2})} \chi_{1}(r) \log^{\alp-1} (\frac{1}{\Omg}) (\rd_{v} \phi)^{2} \\
	\leq & C \bb( \int_{\Gmm^{(1)}_{\tau_{1}}} \big( 1 + \chi_{1}(r) \log^{\alp}(\frac{1}{\Omg}) \big) (\rd_{v} \phi)^{2}
	+ \int_{\Gmm^{(2)}_{\tau_{1}}} (\rd_{u} \phi)^{2} \bb)	.
\end{aligned}
\end{equation}
\end{proposition}

\begin{proof} 
When $\alp = 0$, \eqref{eq:intr:red-shift} follows from \eqref{eq:intr:energy}. Hence it suffices to consider the case $\alp > 0$. We begin with
 \begin{align}
	0 =& \iint_{\calD(\tau_{1}, \tau_{2})} \chi_{1}(r) (\log^{\alp} \frac{1}{\Omg}) \rd_{v} \phi 
				\bb( \rd_{u} \rd_{v} \phi - \frac{\Omg^{2}}{r} \rd_{u} \phi - \frac{\Omg^{2}}{r} \rd_{v} \phi \bb) \notag \\
	=& \frac{1}{2} \int_{\Gmm^{(1)}_{\tau_{1}}} \chi_{1}(r) (\log^{\alp} \frac{1}{\Omg}) (\rd_{v} \phi)^{2}
		- \frac{1}{2} \int_{\Gmm^{(1)}_{\tau_{2}}} \chi_{1}(r) (\log^{\alp} \frac{1}{\Omg}) (\rd_{v} \phi)^{2} \label{eq:intr:red-shift:1} \\
	& - \frac{\alp}{4} \iint \chi_{1}(r) (\log^{\alp-1} \frac{1}{\Omg}) \frac{(- \rd_{u} \Omg^{2})}{\Omg^{2}} (\rd_{v} \phi)^{2} \label{eq:intr:red-shift:2}\\
	&  + \frac{1}{2} \iint \Omg^{2} \chi'_{1}(r) (\log^{\alp} \frac{1}{\Omg}) (\rd_{v }\phi)^{2} \label{eq:intr:red-shift:3}\\
	 &  - \iint \chi_{1}(r) \frac{\Omg^{2}}{r} (\log^{\alp} \frac{1}{\Omg}) (\rd_{v} \phi) (\rd_{u} \phi + \rd_{v} \phi). \label{eq:intr:red-shift:4}
\end{align}

Note that
\begin{equation*}
	- \rd_{u} \Omg^{2} = \rd_{u} (1- \frac{2M}{r} + \frac{e^{2}}{r^{2}} ) = - 2 \frac{\Omg^{2}}{r^{2}} (M - \frac{e^{2}}{r}).
\end{equation*}
The crucial observation here is that
\begin{equation*}
	M - \frac{e^{2}}{r_{-}} < 0
\end{equation*}
and hence by choosing $r^{(1)} > r_{-}$ to be sufficiently close to $r_{-}$, so that the support of $\chi_{1}$ is close enough to $\CH$, we have $- \rd_{u} \Omg^{2} \geq c \Omg^{2}$ for some $c > 0$ on the support of $\chi_{1}$ and the space-time integral in \eqref{eq:intr:red-shift:2} has the same sign as the boundary integral on $\Gmm^{(1)}_{\tau_{2}}$ in \eqref{eq:intr:red-shift:1}. 

The term \eqref{eq:intr:red-shift:3} can be controlled using Proposition~\ref{prop:intr:ILED}, as it is safely localized away from $\CH$. Finally, for \eqref{eq:intr:red-shift:4}, we first use the Cauchy-Schwarz inequality to write
\begin{equation*}
	\abs{\eqref{eq:intr:red-shift:4}}
	\leq \frac{\varepsilon}{2} \iint \chi_{1}(r) (\log^{\alp-1} (\frac{1}{\Omg})) (\rd_{v} \phi)^{2} + \frac{1}{2 \varepsilon} \iint \chi_{1}(r) \, \frac{\Omg^{4}}{r^{2} \log (\frac{1}{\Omg})} (\rd_{u} \phi + \rd_{v} \phi)^{2} .
\end{equation*}
Choosing $\varepsilon > 0$ sufficiently small (say $\varepsilon = c\alp/4$), the first term can be bounded by \eqref{eq:intr:red-shift:2}, and the second term can be bounded using Proposition~\ref{prop:intr:ILED}. Combining all these estimates leads to the desired conclusion. \qedhere
\end{proof}

Our next proposition is an analogue of Proposition~\ref{prop:intr:red-shift}, capturing the red-shift effect along the event horizon $\calH^{+}$ as we approach $i^{+}$. As in the previous case, we begin by defining
\begin{equation*}
	\chi_{2}(r) = \left\{
\begin{array}{cl}	
	1 & \hbox{ for } r^{(2)} \leq r \leq r_{+} \\
	0 & \hbox{ for } r \leq r^{(2)} - (r_{+} - r^{(2)}) 
\end{array} \right.
\end{equation*}
where $r^{(2)} < r_{+} $ is to be specified below.

\begin{proposition} \label{prop:intr:red-shift:EH}
For $1 \leq \tau_{1} \leq \tau_{2}$, we have
\begin{equation} \label{eq:intr:red-shift:EH}
	\int_{\Gmm^{(2)}_{\tau_{2}}} \chi_{2}(r) \Omg^{-2}(\rd_{u} \phi)^{2}
	+ \int_{\calD(\tau_{1}, \tau_{2})} \chi_{2}(r) \Omg^{-2} (\rd_{v} \phi)^{2} 
	\leq  C \bb( \int_{\Gmm^{(1)}_{\tau_{1}}} (\rd_{v} \phi)^{2}
	+ \int_{\Gmm^{(2)}_{\tau_{1}}} \Omg^{-2} (\rd_{u} \phi)^{2} \bb)
\end{equation}
for some $C>0$ depending only on the parameters $M$ and $e$ of the spacetime.
\end{proposition}

\begin{proof} 
 We begin with
 \begin{align}
	0 =& \iint_{\calD(\tau_{1}, \tau_{2})} \chi_{2}(r) \Omg^{-2} \rd_{u} \phi 
				\bb( \rd_{u} \rd_{v} \phi - \frac{\Omg^{2}}{r} \rd_{u} \phi - \frac{\Omg^{2}}{r} \rd_{v} \phi \bb) \notag \\
	=& 
		-\frac{1}{2} \int_{\Gmm^{(2)}_{\tau_{1}}} \chi_{2}(r) \Omg^{-2} (\rd_{u} \phi)^{2}
		+ \frac{1}{2} \int_{\Gmm^{(2)}_{\tau_{2}}} \chi_{2}(r) \Omg^{-2} (\rd_{u} \phi)^{2} 
				\label{eq:intr:red-shift:EH:1} \\
	& + \frac{1}{2} \iint \chi_{2}(r) \frac{\rd_{v} \Omg^{2}}{\Omg^{4}} (\rd_{u} \phi)^{2} \label{eq:intr:red-shift:EH:2}\\
	&  + \frac{1}{2} \iint \chi'_{2}(r) (\rd_{u}\phi)^{2} \label{eq:intr:red-shift:EH:3}\\
	 &  - \iint \frac{1}{r} \chi_{2}(r) (\rd_{u} \phi) (\rd_{u} \phi + \rd_{v} \phi). \label{eq:intr:red-shift:EH:4}
\end{align}

In this case, we see that if $r^{(2)}$ is chosen sufficiently close to $r_{+}$, then we have
\begin{equation*}
	\rd_{v} \Omg^{2} = - \rd_{v} (1 - \frac{2M}{r} + \frac{e^{2}}{r^{2}} ) = 2 \frac{\Omg^{2}}{r^{2}} (M - \frac{e^{2}}{r^{2}}) \geq c \Omg^{2}
\end{equation*}
 on the support of $\chi_{2}$. This inequality follows from the observation that $M - \frac{e^{2}}{r_{+}} > 0$.

The rest of the proof proceeds similarly to that of Proposition~\ref{prop:intr:red-shift}. Indeed, \eqref{eq:intr:red-shift:EH:3} can be bounded by Proposition~\ref{prop:intr:ILED}, and \eqref{eq:intr:red-shift:EH:4} may be estimated by the Cauchy-Schwarz inequality, \eqref{eq:intr:red-shift:EH:2} and Proposition~\ref{prop:intr:ILED}. \qedhere 
\end{proof}

We now state a proposition which combines all the bounds that have been proved so far.
\begin{proposition} \label{prop:intr:red-shift:final}
For every $\alp \geq 0$, there exists a constant $C = C(\alp) > 0$ such that for $1 \leq \tau_{1} \leq \tau_{2}$, we have
\begin{equation} \label{eq:intr:red-shift:final}
\begin{aligned}
& \hskip-2em
	\int_{\Gmm_{\tau_{2}}^{(1)}} \bb( 1 + \chi_{1}(r) \log^{\alp}(\frac{1}{\Omg}) \bb) (\rd_{v} \phi)^{2} + \int_{\Gmm_{\tau_{2}}^{(2)}} \Omg^{-2} (\rd_{u} \phi)^{2} 
	 + \int_{\CH(\tau_{1}, \tau_{2})}  (\rd_{u} \phi)^{2} + \int_{\EH(\tau_{1}, \tau_{2})} (\rd_{v} \phi)^{2}  \\
	&	+ \iint_{\calD(\tau_{1}, \tau_{2})}  \bb( \Omg^{2} + \alp \chi_{1}(r) \log^{\alp-1}(\frac{1}{\Omg}) \bb) (\rd_{v} \phi)^{2}  
		+	\bb( \Omg^{2} + \chi_{2}(r) \Omg^{-2} \bb) (\rd_{u} \phi)^{2}  \\
	& \leq C \bb( \int_{\Gmm^{(1)}_{\tau_{1}}} \bb( 1 + \chi_{1}(r) \log^{\alp}(\frac{1}{\Omg}) \bb) (\rd_{v} \phi)^{2}
				+ \int_{\Gmm^{(2)}_{\tau_{1}}} \Omg^{-2} (\rd_{u} \phi)^{2} \bb).
\end{aligned}
\end{equation}
\end{proposition}

Iterating Proposition~\ref{prop:intr:red-shift:final}, we obtain a decay statement for $\rd_{v} \phi$ on $\EH$.
\begin{proposition} \label{prop:intr:decay}
If $\phi$ is a solution to the wave equation with spherically symmetric data satisfying \eqref{data.1}-\eqref{data.2} and moreover \eqref{contra.2} holds for some integer $\alp_0> 7$, then for $\tau \geq 1$ we have
\begin{equation} \label{eq:intr:decay}
\begin{aligned}
& \hskip-2em
	\int_{\EH(\tau, \infty)} (\rd_{v} \phi)^{2} 
	& \leq C \tau^{- \alp_{0}}
\end{aligned}
\end{equation}
for some $C=C(A,D,\alp_0)>0$.
\end{proposition}

\begin{proof} 
We will prove the following statement for $n = \alp_{0}$ by an induction on $n$:
\begin{equation} \label{eq:intr:decay:indHyp}
\begin{aligned}
& \hskip-2em
	\tau^{n} \int_{\EH(\tau, \infty)} (\rd_{v} \phi)^{2} 
	+ \sum_{j=0}^{n} (\alp_{0} - j) \tau^{j} \iint_{\calD(\tau, \infty)} \chi_{1}(r) \log^{\alp_{0}-j-1} (\frac{1}{\Omg}) (\rd_{v} \phi)^{2} \\
	& + \tau^{n} \iint_{\calD(\tau, \infty)} \chi_{2}(r) \Omg^{-2} (\rd_{u} \phi)^{2} 
	+ \tau^{n} \iint_{\calD(\tau, \infty)} \Omg^{2} \bb( (\rd_{v} \phi)^{2} + (\rd_{u} \phi)^{2} \bb)
	\leq \calI_{n}
\end{aligned}
\end{equation}
where $\calI_{n}$ is a positive constant depending only on $A$, $D$ and $n$. Observe that the case $n = 0$ is an immediate consequence of Proposition~\ref{prop:intr:red-shift:final}. 

Assume, for the purpose of induction, that \eqref{eq:intr:decay:indHyp} holds for $n = 0, 1, \ldots, n_{0}-1$, where $n_{0}$ is an integer such that $1 \leq n_{0} \leq \alp_{0}$. Then for every $k \in \bbN$, by the pigeonhole principle, there exists $\tau_{k} \in [2^{k}, 2^{k+1}]$ such that
\begin{equation} \label{eq:intr:decay:pigeonhole}
\begin{aligned}
& \hskip-2em
	\int_{\Gmm^{(1)}_{\tau_{k}}} \bb( \big( \Omg^{2} + (\alp_{0} - n_{0} + 1) \chi_{1}(r) \log^{\alp_{0} - n_{0}} (\frac{1}{\Omg}) \big) (\rd_{v} \phi)^{2} 
	+ \Omg^{2} (\rd_{u} \phi)^{2} \bb) \\
	& + \int_{\Gmm^{(2)}_{\tau_{k}}} \bb( (\Omg^{2} + \chi_{2}(r) \Omg^{-2}) (\rd_{u} \phi)^{2} + \Omg^{2} (\rd_{v} \phi)^{2} \bb)
	\leq C \calI_{n_{0}-1} \tau_{k}^{-n_{0}},
\end{aligned}
\end{equation}
for some absolute constant $C > 0$. Observe that the right-hand side of \eqref{eq:intr:red-shift:final} for $\alp = \alp_{0} - n_{0}$ and $\tau_{1} = \tau_{k}$ is bounded by a constant multiple of the left-hand side of \eqref{eq:intr:decay:pigeonhole}, where the constant may depend on $n_{0}$ but is independent of $\tau_{k}$.
By appealing to Proposition~\ref{prop:intr:red-shift:final}, we conclude conclude that \eqref{eq:intr:decay:indHyp} holds with $n = n_{0}$, $\tau = 2^{k}$ and $(\tau, \infty)$ replaced by the interval $[\tau_{k}, 4 \tau_{k}]$. Note, in particular, that $[2^{k+1}, 2^{k+2}] \subseteq [\tau_{k}, 4 \tau_{k}]$. Summing up these bounds for $k \geq 0$, and using the trivial bound from Proposition~\ref{prop:intr:red-shift:final} for the interval $\tau \in [1, 2]$, we obtain \eqref{eq:intr:decay:indHyp} for $n = n_{0}$, as desired. \qedhere
\end{proof}

We now conclude the proof of Theorem \ref{final.blow.up.step}:
\begin{proof}[Proof of Theorem \ref{final.blow.up.step}]
Using Proposition \ref{prop:intr:decay}, we in particular have 
$$\int_{\mathcal H^+(\tau,2\tau)} \tau^{\alp_0} (\rd_v\phi)^2 \leq C.$$
We apply this estimate for a sequence $\tau_k=2^k$ to get
\begin{equation*}
\begin{split}
\int_{\EH\cap\{v\geq 1\}} \f{v^{\alp_0}}{\log^2(1+v)}(\rd_v\phi)^2\leq &C\sum_{k=0}^\infty \int_{\mathcal H^+(\tau_k,\tau_{k+1})} \f{\tau^{\alp_0}}{(k+1)^2}(\rd_v\phi)^2\leq C\sum_{k=0}^{\infty}\f{1}{(k+1)^2}<\infty.
\end{split}
\end{equation*}
We have thus achieved \eqref{contra.2.goal} and conclude the proof of Theorem \ref{final.blow.up.step}.
\qedhere
\end{proof}

\section{Construction of blow up solution}\label{sec.construction}
In this section, we work in the exterior of the black hole.
The purpose of this section is to demonstrate the existence of a linear wave with smooth and compactly supported data on $\uC_{1} \cup C_{-\infty}$ such that $\mathfrak{L} \neq 0$, i.e., we prove Theorem \ref{construction.existence.thm}.

Consider a smooth function $\chi$ that satisfies 
\begin{equation} \label{eq:pert:chi}
	\mathrm{supp} \, \chi \subseteq [0, 1], \quad \sup \abs{\chi} \leq 1, \quad \int_0^1 \chi(u) \ud u = 0,
\end{equation}
as well as
\begin{equation} \label{eq:pert:chi:lb}
	\int_{0}^{1} \int_{0}^{u} \chi(u') \ud u' \ud u \geq \frac{1}{8}.
\end{equation}
Indeed, such a function $\chi$ can be easily constructed by mollifying a step function which takes the values $1$ and $-1$ in $(\eps, 1/2)$ and $(1/2, 1-\eps)$, respectively, and equals zero everywhere else.

Let $\phi$ be the solution to the linear wave equation with data on $\uC_{1} \cup C_{-\Ufree}$ given by
\begin{equation} \label{eq:pert:id}
	\rd_{u}(r \phi)(u, 1) = \chi(u + \Ufree), \quad
	\phi(-\Ufree, v) = 0.
\end{equation}
where $\Ufree > 0$ is a large constant to be chosen\footnote{Unlike the other parts of the paper, where $U_{0}$ is a fixed number, in this section $U_{0}$ is a large parameter used in the construction.}.

Observe that, thanks to the condition $\int \chi = 0$, we have
\begin{equation*}
	\phi(u, 1) = 0 \quad \hbox{ for } u \not \in [-\Ufree, -\Ufree + 1],
\end{equation*}
and hence $\phi$ has compact support on the initial curve $\uC_{1} \cup C_{-\Ufree}$.

The main result of this section reads as follows.
\begin{theorem} \label{thm:pert}
For $\Ufree > 0$ sufficiently large, we have $\mathfrak{L} \neq 0$ for the solution $\phi$ to the linear wave equation \eqref{eq:pert:wave} with initial data \eqref{eq:pert:id}.
\end{theorem}

We begin by fixing a number $\Ufixed > 0$, which is chosen so that
\begin{equation} \label{eq:pert:cond4U0}
\begin{gathered}
	\frac{1}{\sqrt{2}} r \leq 1 - u \leq \sqrt{2} \, r \quad \hbox{ on } \uC_{1} \cap \set{u \leq - \Ufixed}, \\
	\Ufixed \geq 100 M.
\end{gathered}
\end{equation}
We emphasize that $\Ufixed$ is chosen depending only on the parameters $M$, $e$ of the Reissner-Nordstr\"om background. In what follows, $\Ufree > \Ufixed$ will be chosen sufficiently large so that the data for $\phi$ on $C_{-\Ufixed}$ is small enough in an appropriate sense. This procedure effectively reduces the problem of computing $\mathfrak{L}$ on the whole null infinity $\calI^{+}$ to a finite $u$ region $\calI^{+} \cap \set{-\Ufree \leq u \leq -\Ufixed}$. In fact, we will show that $r \phi$, hence the integrand in the definition of  $\mathfrak{L}$, is largest in the region $\set{-\Ufree \leq u \leq -\Ufree+1}$, thanks to our choice of initial data. A graphical summary of our strategy is given in Figure~\ref{fig:pert}.

\begin{figure}[h]
\begin{center}
\def\svgwidth{220px}
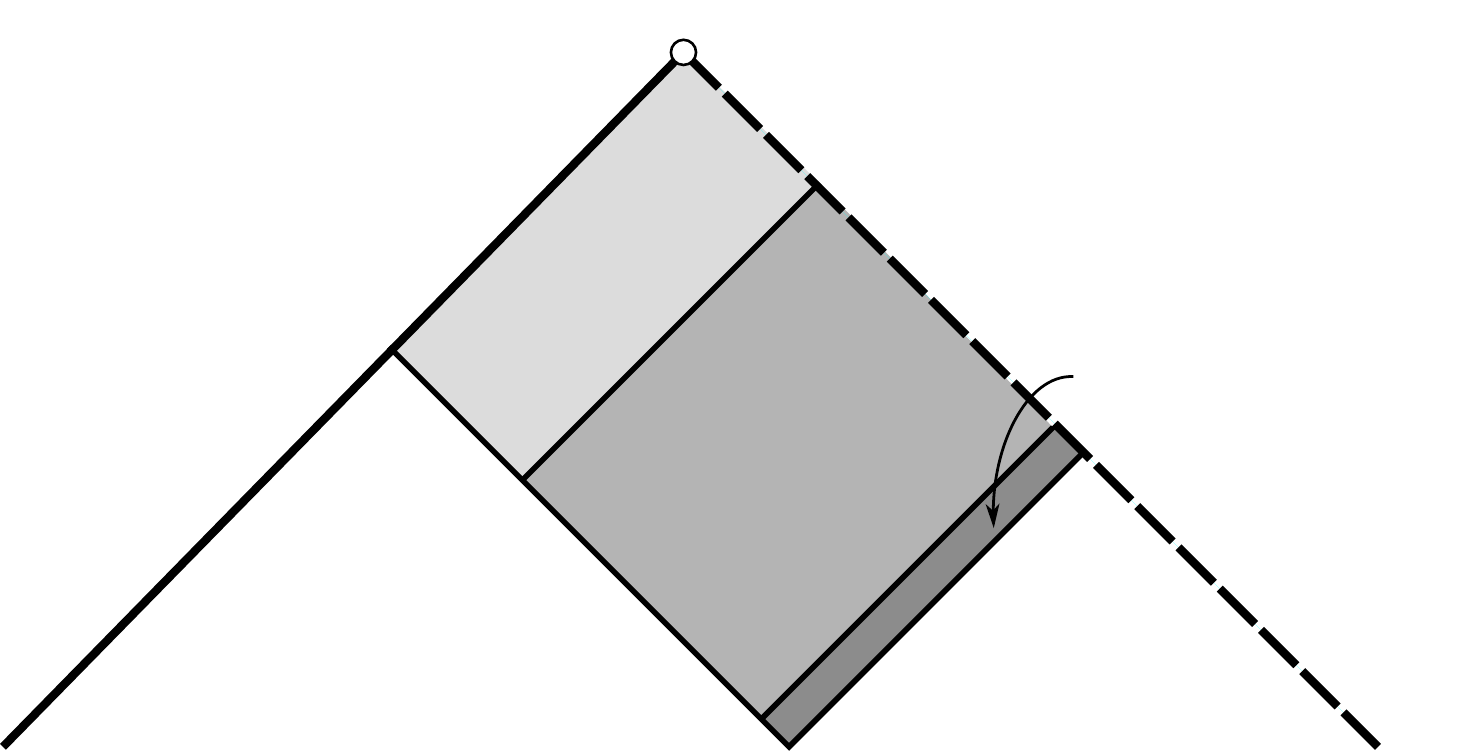 
\caption{} \label{fig:pert}
\end{center}
\end{figure}

Before we begin the proof, a word on the form of the wave equation is in order. In the region where $r$ is large, it is convenient to use the following formulation of the linear wave equation (c.f. Section \ref{sec.lower.bd.R}):
\begin{equation} \label{eq:pert:wave}
	\rd_{u} \rd_{v} (r \phi) = - \frac{2 (M - \frac{e^{2}}{r}) \Omg^2}{r^{2}} \phi.
\end{equation}
If we work in the exterior of the black hole, i.e., where $r > r_{+}$, then $M - \frac{e^{2}}{r} > 0$, and hence $\abs{M - \frac{e^{2}}{r}} \leq M$. In the rest of this section, this observation will be used frequently.

We begin by analyzing the behavior of $\phi$ in the region $\set{-\Ufree \leq u \leq -\Ufree+1}$.
\begin{proposition} \label{prop:pert:-U}
Let $\Ufree > \Ufixed$ be sufficiently large, depending on $M$ and $e$. Then for $u \in [-\Ufree, -\Ufree+1]$, we have
\begin{align}
	\abs{\rd_{u} (r \phi)(u, v) - \chi(u + \Ufree)} \leq& \frac{4M}{\Ufree^{2}}, \label{eq:pert:-U:durphi}\\
	\abs{r \phi}\leq & \frac{3}{2}. \label{eq:pert:-U:rphi}
\end{align}
Moreover, on the null curve $C_{-\Ufree+1}$, we have
\begin{align} 
	\abs{r \phi(-\Ufree+1, v)} \leq & \frac{4M}{\Ufree^{2}}, \label{eq:pert:-U+1:rphi} \\
	\abs{\rd_{v} (r \phi)(-\Ufree+1, v)} \leq & \frac{3M}{r^{3}(-\Ufree+1, v)}. \label{eq:pert:-U+1:dvrphi}
\end{align}
\end{proposition}
\begin{proof} 
We begin by closing the bounds \eqref{eq:pert:-U:durphi} and \eqref{eq:pert:-U:rphi}. Assume, for the purpose of a bootstrap argument, that we have \eqref{eq:pert:-U:durphi} and \eqref{eq:pert:-U:rphi} when $v \in [1, V]$ for some $V > 1$ and $u \in [-\Ufree, -\Ufree+1]$, where $\Ufree > 0$ is to be determined. 

Integrating \eqref{eq:pert:-U:durphi} in $u$ from the initial condition $r \phi = 0$ on $C_{-\Ufree}$ and using \eqref{eq:pert:chi}, we obtain
\begin{equation} \label{eq:pert:-U:rphi:imp}
	\abs{r \phi} \leq 1 + \frac{4M}{\Ufree^{2}},
\end{equation}
which improves upon \eqref{eq:pert:-U:rphi} for, say, $\Ufree^{2} \geq 10 M$. Hence it remains to establish an improvement over \eqref{eq:pert:-U:durphi}. Integrating \eqref{eq:pert:wave} in $v$, we have
\begin{equation*}
	\abs{\rd_{u}(r \phi)(u, v) - \chi(u+\Ufree)} \leq  \int_{1}^{v} \frac{2M \Omg^2}{r^{3}}\,  \abs{r \phi} (u, v') \, \ud v'.
\end{equation*}
Recalling that $\rd_{v} r = \Omg^2$, we can estimate the right-hand side uniformly in $v \in [1, \infty)$ as follows:
\begin{align*}
	\sup_{v \in [1, \infty)} \int_{1}^{v} \frac{2M \Omg^2}{r^{3}}\, \abs{ r \phi }(u, v') \, \ud v'
	\leq & \frac{3}{2} \int_{r(u, 1)}^{\infty} \frac{2M}{r^{3}} \, \ud r 
	= \frac{3}{2} \frac{M}{r(u, 1)^{2}}
\end{align*}
As $\Ufree > \Ufixed$, we have $r(u, 1)^{2} \geq \frac{1}{2} (1-u)^{2} \geq \frac{1}{2} \Ufree^{2}$ for $u \in [-\Ufree, -\Ufree+1]$ by \eqref{eq:pert:cond4U0}. Hence we obtain the following improvement over \eqref{eq:pert:-U:durphi}:
\begin{equation} \label{eq:pert:-U:durphi:imp}
	\abs{\rd_{u}(r \phi)(u, v) - \chi(u+\Ufree)} \leq \frac{3M}{\Ufree^{2}}.
\end{equation}
for every $u \in [-\Ufree, -\Ufree+1]$ and $v \in [1, V]$. By a standard bootstrap argument, \eqref{eq:pert:-U:durphi} and \eqref{eq:pert:-U:rphi} follow for all $u \in [-\Ufree, -\Ufree+1]$ and $v \in [1, \infty)$. 

We are ready to conclude the proof. The bound \eqref{eq:pert:-U+1:rphi} follows from integrating \eqref{eq:pert:-U:durphi} on $u \in [-\Ufree, -\Ufree+1]$, exploiting the fact that $\int_{-\Ufree}^{-\Ufree+1} \chi(u+\Ufree) \, \ud u = 0$. Finally, to prove \eqref{eq:pert:-U+1:dvrphi}, we integrate \eqref{eq:pert:wave} in $u$ to write
\begin{align*}
	\rd_{v}(r \phi)(-\Ufree+1, v) = - \int_{-\Ufree}^{-\Ufree+1} \frac{2(M-\frac{e^{2}}{r}) \Omg^2}{r^{3}} \, r \phi(u, v) \, \ud u
\end{align*}
Using \eqref{eq:pert:-U:rphi} and the fact that $\Omg^2 \leq 1$, the desired estimate follows. \qedhere
\end{proof}

Next, we study the behavior of $\phi$ in the region $\set{-\Ufree +1 \leq u \leq -\Ufixed}$. In particular, we show that the data for $\phi$ on the curve $C_{-\Ufixed}$ can be made small by taking $\Ufree$ to be sufficiently large.
\begin{proposition} \label{prop:pert:-U0}
Let $\Ufree > \Ufixed$ be sufficiently large, depending on $M$ and $e$. Then for $u \in [-\Ufree+1, - \Ufixed]$ and $v \in [1, \infty)$, we have 
\begin{equation} \label{eq:pert:-U0:rphi}
	\abs{r \phi(u,v)} \leq  \frac{5 M}{\Ufree^{2}}. 
\end{equation}
Moreover, for $(u, v) \in \set{- \Ufixed} \times [1, \infty)$, we have
\begin{equation} \label{eq:pert:-U0:dvrphi}
	\abs{\rd_{v} (r \phi)(u, v)} \leq \frac{3M}{r^{3}(-\Ufree+1, v)} + \frac{10 M^{2}}{\Ufree r^{3}(u, v)}.
\end{equation}
\end{proposition}
\begin{proof} 
Fix $V > 1$ and consider $(u, v) \in [-\Ufree+1, -\Ufixed] \times [1, V]$. Integrating \eqref{eq:pert:wave} on $v' \in [1, v]$ and using \eqref{eq:pert:cond4U0}, we have
\begin{equation*}
	\abs{\rd_{u} (r \phi)(u, v)} \leq \int_{1}^{v} \frac{2M\Omg^2}{r^{3}} (r \phi)(u, v') \, \ud v'
	\leq \sup \abs{r \phi} \frac{M}{r^{2}(u, 1)}
	\leq  \sup \abs{r \phi} \frac{2 M}{(1-u)^{2}},
\end{equation*}
where the supremum is taken over $u \in [-\Ufree+1, -\Ufixed]$ and $v \in [1, V]$. Integrating from \eqref{eq:pert:-U+1:rphi}, we then obtain
\begin{align*}
	\abs{r \phi(u,v)} 
	\leq  \bb( \int_{-\Ufree+1}^{-\Ufixed} \frac{2M}{(1-u')^{2}} \, \ud u' \bb) \sup \abs{r \phi} + \frac{4M}{\Ufree^{2}} 
	\leq \frac{1}{50} \sup\abs{r \phi} + \frac{4M}{\Ufree^{2}}.
\end{align*}
where we used $\Ufixed \geq 100 M$ in \eqref{eq:pert:cond4U0} for the last inequality. Hence, by a standard bootstrap argument, we obtain \eqref{eq:pert:-U0:rphi} for $u \in [-\Ufree+1, -\Ufixed]$ and $v \in [1, \infty)$.

To prove \eqref{eq:pert:-U0:dvrphi}, we begin by integrating \eqref{eq:pert:wave} on $u \in [-\Ufree+1, -\Ufixed]$ to obtain
\begin{align*}
	\rd_{v}(r \phi)(-\Ufixed, v) = \rd_{v}(r \phi)(-\Ufree+1, v) - \int_{-\Ufree+1}^{-\Ufixed} \frac{2(M-\frac{e^{2}}{r}) \Omg^2}{r^{3}} \, r \phi(u, v) \, \ud u.
\end{align*}
Using Proposition~\ref{prop:pert:-U} for the first term and \eqref{eq:pert:-U0:rphi} for the second, \eqref{eq:pert:-U0:dvrphi} follows. \qedhere
\end{proof}

Finally, in order to obtain Theorem~\ref{thm:pert}, we put together the previous estimates and the decay theorem of Dafermos-Rodnianski (Theorem \ref{thm:DRPL}):

\begin{proof} [Proof of Theorem~\ref{thm:pert}]
By our choice of initial data, we have $\lim_{v \to \infty} r^{3} \rd_{v}(r \phi)(-\Ufree, v) = 0$. 
In order to show that $\mathfrak{L} \neq 0$, it suffices by the dominated convergence theorem to show that
 \begin{equation} \label{eq:pert:goal}
	\lim_{v \to \infty} \int_{-\Ufree}^{\infty} r \phi(u, v) \, \ud u \neq 0.
\end{equation}
We control the integral on the left hand side for large but finite $v$, restricted to the region $r\geq R_0$, where $R_0$ is as given in Theorem \ref{thm:DRPL}.  We split this integral into three integrals:
\begin{align*}
	\int_{-\Ufree}^{u_{R_0^*}(v)} r \phi(u, v) \, \ud u  
	&=  \int_{-\Ufree}^{-\Ufree+1} r \phi(u, v) \, \ud u + \int_{-\Ufree+1}^{-\Ufixed} r \phi(u, v) \, \ud u + \int_{-\Ufixed}^{u_{R_0^*}(v)} r \phi(u, v) \, \ud u \\
	&=: \mathrm{I}_{1}(v) + \mathrm{I}_{2}(v) + \mathrm{I}_{3}(v) \, .
\end{align*}
Here, we use the convention as before the $u_{R_0^*}(v)$ denotes the unique $u$ value such that $r(u,v)=R_0$. The integral $\mathrm{I}_{1}$ constitutes the main contribution, whereas $\mathrm{I}_{2}$ and $\mathrm{I}_{3}$ are errors. As emphasized at the beginning of this section, we have the freedom to choose $\Ufree$ as large as we need to reduce the size of the error terms, but $\Ufixed$ \emph{needs to remain fixed} in order to apply Theorem~\ref{thm:DRPL}.
 
We begin with the contribution of $\mathrm{I}_{1}(v)$. Since $r \phi(-\Ufree, 1) = 0$, we may write
\begin{equation*}
	\mathrm{I}_{1}(v) = \int_{-\Ufree}^{-\Ufree+1} \int_{-\Ufree}^{u} \rd_{u} (r \phi) (u', v) \, \ud u' \, \ud u.
\end{equation*} 
Then by \eqref{eq:pert:-U:durphi} in Proposition~\ref{prop:pert:-U}, we have
\begin{equation*} 
	\abs{\mathrm{I_{1}}(v) - \int_{-\Ufree}^{-\Ufree+1}\int_{-\Ufree}^{u}\chi(u'+\Ufree) \, \ud u' \, \ud u} \leq \int_{-\Ufree}^{-\Ufree+1} \int_{-\Ufree}^{u} \frac{4M}{\Ufree^{2}} \, \ud u' \, \ud u
	\leq \frac{2M}{\Ufree^{2}}.
\end{equation*}
By \eqref{eq:pert:chi:lb}, it therefore follows that we have the lower bound
\begin{equation} \label{eq:pert:I1}
	 \mathrm{I_{1}}(v) \geq \frac{1}{8} - \frac{2M}{\Ufree^{2}} 
\end{equation}
for every $v \in [1, \infty)$. Observe that the error $\f{2M}{\Ufree^2}$ can be made as small as we wish by taking $\Ufree$ sufficiently large.

Next, by \eqref{eq:pert:-U0:rphi} in Proposition~\ref{prop:pert:-U0}, for $v \in [1, \infty)$ we immediately have
\begin{equation} \label{eq:pert:I2}
	\abs{\mathrm{I_{2}}(v)} \leq \int_{-\Ufree+1}^{-\Ufixed} \frac{5M}{\Ufree^{2}} \, \ud u \leq \frac{5M}{\Ufree},
\end{equation}
which again can be made arbitrarily small by taking $\Ufree \to \infty$.

Finally, by \eqref{eq:pert:-U0:rphi}, \eqref{eq:pert:-U0:dvrphi} and the fact that $\rd_{v} r = \Omg^2$ is bounded from below by a positive constant on $C_{-\Ufixed} \cap \set{v \geq 1}$ (depending only on $\Ufixed$), for any $\omg < 3$ we have
\begin{equation*}
	\sup_{v \in [1, \infty)} r^{2} \abs{\frac{\rd_{v} \phi}{\rd_{v} r}}(-\Ufixed, v) + r^{\omg} \abs{\frac{\rd_{v}(r \phi)}{\rd_{v} r}} (-\Ufixed, v)  \to 0 \quad \hbox{ as } \Ufree \to \infty.
\end{equation*}
In other words, $\phi$ satisfies \eqref{eq:DRPL:hyp:C} for any $1 < \omg < 3$ with $A = A(\Ufree) \to 0$ as $\Ufree \to \infty$. Moreover, \eqref{eq:DRPL:hyp:uC} is satisfied for any $A > 0$, as $\rd_{u}(r \phi) = \phi = 0$ on $\uC_{1} \cap \set{u \geq -\Ufixed}$. Using Theorem~\ref{thm:DRPL} with $\omg = 5/2$, it follows that we have
\begin{equation} \label{eq:pert:I3}
	\lim_{v \to \infty} \abs{\mathrm{I_{3}}(v)} \to 0 \quad \hbox{ as } \Ufree \to \infty.
\end{equation}

Putting together \eqref{eq:pert:I1}, \eqref{eq:pert:I2} and \eqref{eq:pert:I3}, it follows that $\mathfrak{L} = \lim_{v \to \infty} \sum_{j=1}^{3} \mathrm{I}_{j} > \frac{1}{16}$ if $\Ufree$ is chosen sufficiently large, which concludes the proof. \qedhere
\end{proof}

\bibliographystyle{hplain}
\bibliography{SCClinear}
\end{document}